\documentclass[11pt, reqno]{amsart}
\usepackage{amssymb, amscd, latexsym}
\usepackage{amsmath}
\usepackage{amsthm}
\usepackage{graphicx}
\usepackage[all]{xy}
\usepackage{float}
\usepackage{mathrsfs}
\usepackage{pgf,tikz}
\usetikzlibrary{arrows}
\usepackage{array}
\usepackage{arydshln}
\usepackage{hyperref}
\usepackage{enumitem}
\usepackage{graphicx} 

\makeatletter
\newcommand{\addresseshere}{%
  \enddoc@text\let\enddoc@text\relax
}
\makeatother

\theoremstyle{plain}

\newtheorem{thm}{Theorem}[section]
\newtheorem{prop}[thm]{Proposition}
\newtheorem{cor}[thm]{Corollary}

\newtheorem{conj}[thm]{Conjecture}

\theoremstyle{definition}
\newtheorem{defn}[thm]{Definition}
\newtheorem{rmk}[thm]{Remark}

\newtheorem{notat}[thm]{Notation}

\newtheorem{ex}[thm]{Example}

\newtheorem{ass}[thm]{Assumption}

\newcommand{\F}{{\mathcal F}}
\newcommand{\FF}{{\mathbb F}}

\newcommand{\NN}{{\mathbb N}}

\newcommand{\ZZ}{{\mathbb Z}}

\DeclareMathOperator{\reg}{reg}

\DeclareMathOperator{\sd}{solv.deg}

\DeclareMathOperator{\ttop}{top}
\DeclareMathOperator{\maxgb}{max.GB.deg}

\title{Semi-regular sequences and other random systems of equations}

\author[Bigdeli]{M. Bigdeli}
\address{School of Mathematics, Institute for Research in Fundamental Sciences (IPM),  P.O.Box: 19395-5746, Teheran, Iran}
\email{mina.bigdeli98@gmail.com, mina.bigdeli@ipm.ir}

\author[De Negri]{E. De Negri}
\address{Dipartimento di Matematica, Universit\`a di Genova, Via Dodecaneso 35, 16146 Genova, Italy}
\email{denegri@dima.unige.it}

\author[Dizdarevic]{M. M. Dizdarevic}
\address{Faculty of Natural Sciences and Mathematics, University of Sarajevo, Bosnia and Herzegovina}
\email{manuela@dizdarevic.org}

\author[Gorla]{E. Gorla}
\address{Institut de Math\'ematiques, Universit\'e de Neuch\^atel, Rue Emile-Argand 11, 2000 Neuch\^atel, Switzerland}  
\email{elisa.gorla@unine.ch}

\author[Minko]{R. Minko}
\address{Mathematical Institute, University of Oxford, Oxford, OX1 3BJ, United Kingdom}
\email{romy.minko@wolfson.ox.ac.uk}

\author[Tsakou]{S. Tsakou}
\address{Universit\'e de Picardie Jules Verne, 80039 Amiens,  France}
\email{sulamithe.tsakou@u-picardie.fr}

\thanks{This work was started during the collaborative conference ``Women in Numbers Europe 3''. The authors would like to acknowledge the organizers Sorina Ionica, Holly Krieger, and Elisa Lorenzo Garcia as well as the Henri Lebesgue Center, which hosted the conference. 
The symbolic algebra computations were performed with CoCoA~5~\cite{cocoa}, Macaulay2~\cite{m2}, Magma~\cite{magma}, and Wolfram Mathematica~\cite{mathematica}.}

\begin{document}

\begin{abstract}
The security of multivariate cryptosystems and digital signature schemes relies on the hardness of solving a system of polynomial equations over a finite field. Polynomial system solving is also currently a bottleneck of index-calculus algorithms to solve the elliptic and hyperelliptic curve discrete logarithm problem. The complexity of solving a system of polynomial equations is closely related to the cost of computing Gr{\"o}bner bases, since computing the solutions of a polynomial system can be reduced to finding a lexicographic Gr{\"o}bner basis for the ideal generated by the equations. Several algorithms for computing such bases exist: We consider those based on repeated Gaussian elimination of Macaulay matrices. 
In this paper, we analyze the case of random systems, where random systems means either semi-regular systems, or quadratic systems in $n$ variables which contain a regular sequence of $n$ polynomials. We provide explicit formulae for bounds on the solving degree of semi-regular systems with $m>n$ equations in $n$ variables, for equations of arbitrary degrees for $m=n+1$, and for any $m$ for systems of quadratic or cubic polynomials. In the appendix, we provide a table of bounds for the solving degree of semi-regular systems of $m=n+k$ quadratic equations in $n$ variables for $2\leq k,n\leq 100$ and online we provide the values of the bounds for $2\leq k,n\leq 500$. For quadratic systems which contain a regular sequence of $n$ polynomials, we argue that the Eisenbud-Green-Harris conjecture, if true, provides a sharp bound for their solving degree, which we compute explicitly. 
\end{abstract}
\maketitle

\section*{Introduction}

Cryptosystems and digital signature algorithms based on the hardness of solving systems of multivariate polynomial equations belong to one of the five major families of post-quantum cryptography, multivariate (public-key) cryptography. Many systems arising in this context consist of quadratic equations and are therefore referred to as multivariate quadratic (MQ) systems.  Given a system of quadratic polynomials over a finite field, $\mathcal{F}$, and a vector $y$, the problem of finding a vector $x$ such that $\mathcal{F}(x)=y$, referred to as the MQ problem, is known to be NP-hard. 

Polynomial system solving is also a crucial step in index-calculus algorithms to solve the discrete logarithm problem on an elliptic curve, a hyperelliptic curve, or an abelian variety. In such algorithms, one attempts to find a decomposition of a point of an elliptic curve (or of an abelian variety) over a chosen factor base. Such a decomposition is produced as a solution to a suitable system of polynomial equations. In this setting, most systems considered will not have any solutions and once in a while one will produce a system which has a solution and hence produces a relation. The found relations allow the attacker to set up a linear system, whose solution reveals the discrete logarithm. 


The best known approach for solving an arbitrary system of polynomial equations over a finite field is finding a Gr\"obner basis of the ideal generated by the polynomials in the system \cite{CG17}. The first algorithm for computing Gr\"obner bases was introduced by Buchberger \cite{ buchberger1965grobner} in 1965. Subsequently a number of system-solver algorithms have been proposed, including~\cite{lazard1983grobner, KS99, CKPS00, faugere2002new, bettale2009hybrid}. For a thorough summary, we refer the reader to \cite{joux2009algorithmic}. 

Importantly, these algorithms do not rely on the specific algebraic structure of the original polynomials; they can therefore be applied to any system of equations. In particular, the complexity of computing a Gr\"obner basis of the public key of a multivariate cryptosystem or a multivariate digital signature algorithm gives an upper bound on the security of that system. Consequently, we are motivated to find tighter bounds on the complexity of Gr\"obner basis algorithms.

Several system solvers use Gaussian elimination of Macaulay matrices \cite{joux2009algorithmic} to obtain a Gr\"obner basis. The complexity depends on the size to which this matrix grows, which in turn depends on the solving degree $D$ of the system. This is the degree at which Gaussian elimination on the matrix reveals a Gr\"obner basis. It is difficult to know $D$ in advance, so in practice the complexity of a system solver is often estimated from the degree of regularity of the system. Bardet, Faug\`ere, and Salvy analyzed the F5 Algorithm in~\cite{bardet2004complexity}, giving a definition for the degree of regularity and an asymptotic upper bound on this degree for cryptographic semi-regular sequences. This bound is used widely in the cryptography community to estimate the security of multivariate polynomial cryptosystems and digital signatures.

In Section~\ref{sect:dreg} of this paper, we provide evidence that the degree of regularity does not bound the solving degree of a system, in the case of inhomogeneous systems. Examples of this kind already appear in~\cite{CG17}. In this paper, however, we provide examples for which the difference between the degree of regularity and the solving degree is larger than in previously known examples. 

Often cryptographers make the assumption that the systems that they analyze are random, where random means that the coefficients of the polynomials in $\mathcal{F}$ are chosen uniformly at random from the coefficient field.  Of course, in practice MQ cryptosystems are \textit{not} random in this sense, as they must be equipped with a backdoor with which a trusted user can easily invert the system.  However, this invertible map is hidden by secret affine transformations, such that the resulting public system resembles a random system. In the same way, systems coming from index-calculus are not truly random systems. They, however, more closely resemble random systems, especially since they usually have no solutions, as it is the case for an overdetermined random system of equations.

In Sections~\ref{semi-regular} and~\ref{sec:eisen} of this paper, we consider two distinct mathematical formulations for randomness of a system of equations. The first family of random systems that we consider are (cryptographic) semi-regular systems. Notice that systems whose coefficients are chosen uniformly at random in an infinite field (or in a large enough field) are known to be semi-regular for some choices of the parameters, and conjectured to always be. 
The second family of random systems that we consider are overdetermined quadratic systems of $m$ equations in $n$ variables, which contain a regular sequence of $n$ polynomials. Notice that a system of quadratic equations for which the coefficients of the first $n$ polynomials are chosen uniformly at random over an infinite (or large enough) field belongs to this family (the remaining $m-n$ polynomials can be chosen arbitrarily). Therefore, this is a family of random systems which is larger than the one which is usually studied in the cryptographic literature.

\bigskip
The structure of the paper is organized as follows. In Section \ref{sec:not} we review some mathematical background on commutative algebra and semi-regular sequences that will be needed in the rest of the paper. In Section \ref{sect:dreg} we compare the notions of solving degree, degree of regularity, and Castelnuovo-Mumford regularity and give examples of polynomial systems for which the solving degree is greater than degree of regularity. Subsequently, in Section \ref{semi-regular} we derive explicit bounds on the solving degree of homogeneous systems of $n+1$ equations in $n$ variables, and for systems of $m$ quadratic or cubic equations in $n$ variables, for any $m$. Motivated by the examples in Section~\ref{sect:dreg}, we propose a new definition of cryptographic semi-regular sequence for inhomogeneous systems and we provide bounds for the solving degree of inhomogeneous systems which are cryptographic semi-regular systems according to our definition. Section \ref{sec:eisen} describes how the Eisenbud-Green-Harris conjecture, if true, can be used to bound the solving degree of overdetermined systems of quadratic polynomials. We conclude the paper by discussing some limitations to the applicability of the results of Section~\ref{sec:eisen} to systems arising in cryptography, and a connection to the degree of regularity.
In the appendix, we provide tables of bounds for the solving degree of semi-regular systems of $m=n+k$ quadratic equations in $n$ variables for $2\leq k,n\leq 100$ and online we provide the values of the bounds for $2\leq k,n\leq 500$. 

\section{Notation and preliminaries}\label{sec:not}

In this section we recall some concepts, definitions and results which will be used throughout the paper. 
Let $\mathbb{K}$ be a field and let $R=\mathbb{K}[x_1,\ldots,x_n]$ be the polynomial ring in $n$ variables with coefficients in $\mathbb{K}$. Denote by $\mathrm{Mon}(R)$ the set of monomials of $R$ and consider the degree reverse lexicographic order on $R$. We let $\maxgb(I)$ denote the largest degree of an element in the reduced degree reverse lexicographic Gr\"obner basis of the ideal $I$.
For $a\in\mathbb{R}$ we denote by $\lfloor a\rfloor$ and $\lceil a\rceil$ the floor and the ceiling of $a$, respectively.

Let $\F=\{f_1,\ldots,f_m\}\subseteq R$ be a system of polynomial equations and let $I=(f_1,\ldots,f_m)$ be the ideal that they generate. A system is {\bf overdetermined} if $m>n$. Let $d_i=\deg(f_i)$ for $1\leq i\leq m$. We may assume without loss of generality that $d_i\geq 2$ for all $i$. In fact, if $\F$ contains a polynomial of degree $0$, this is either $0$ and can be eliminated, or an element of $\mathbb{K}\setminus\{0\}$ and the reduced degree lexicographic Gr\"obner basis of $I$ is the polynomial $1$. If $\F$ contains a polynomial of degree $1$, this can be used to produce a new system in one less equation and one less variable, which has the same solutions as $\F$.

If the equations of $\F$ are not homogeneous, then we may associate to $\F$ the system 
$$\F^h=\{f_1^h,\ldots,f_m^h\}.$$ 
Here, for $f\in R$, we denote by $f^h\in S=R[t]$ the homogenization of $f$ with respect to a new variable $t$. If $\F$ is a system of $m$ inhomogeneous equations in $n$ variables, then $\F^h$ is a system of $m$ homogeneous equations in $n+1$ variables. Denote by $J=(\F^h)\subseteq S$ the ideal generated by $\F^h$.

One may associate to $\F$ another homogeneous system, whose equations are obtained from those of $\F$ by dropping the lower degree monomials. Precisely, for $f\in R$ let 
$$f^{\ttop}=f^h(x_1,\ldots,x_n,0)$$ be the polynomial obtained from $f$ by homogenizing it with respect to $t$ and setting $t=0$. In other words, $f^{\ttop}$ is the homogeneous part of $f$ of highest degree. We regard $f^{\ttop}$ as an element of $R$. Then $$\F^{\ttop}=\{f_1^{\ttop},\ldots,f_m^{\ttop}\}$$ is a system of $m$ homogeneous equations in $n$ variables. Denote by $(\F^{\ttop})\subseteq R$ the ideal generated by $\F^{\ttop}$.

For any $f_1,\ldots,f_m\in R$, we may assume without loss of generality that $f_1^{\ttop},\ldots,f_m^{\ttop}$ are linearly independent. In fact, in case they are not, an equivalent system of equations with linearly independent homogeneous parts of highest degree can be obtained in polynomial time from $f_1,\ldots,f_m$ by Gaussian elimination. Notice that if $f_1^{\ttop},\ldots,f_m^{\ttop}$ are linearly independent, then so are $f_1,\ldots,f_m$ and $f_1^h,\ldots,f_m^h$.

For $d\geq 0$, denote by $R_d$ the $\mathbb{K}$-vector space generated by the monomials of $R$ of degree $d$. Then $\dim_{\mathbb{K}} R_d={n+d-1 \choose d}$. 
If $I\subseteq R$, then
$$I_d=\{f\in I\mid f\mbox{ homogeneous of degree $d$}\}\cup\{0\}$$ 
is a finite dimensional $\mathbb{K}$-vector space for all $d\geq 0$. 
Clearly, $(R/I)_d=R_d/I_d$ is also a finite dimensional $\mathbb{K}$-vector space of dimension 
$\dim_{\mathbb{K}} (R/I)_d=\dim_{\mathbb{K}} R_d -\dim_{\mathbb{K}} I_d.$

\subsection{Commutative algebra review}

Let $I\subseteq R$ be a homogeneous ideal and suppose that $f_1,\ldots,f_\mu$ is a minimal system of generators of $I$ with $\deg f_i=d_i$ for all $i$. Then $I$ is the homomorphic image of a free $R$-module $\mathbb{F}_0$. More precisely,  there is an epimorphism $\phi_0:\ \mathbb{F}_0\rightarrow I$, where $\mathbb{F}_0=\oplus_{i=1}^\mu R(-d_i)$ with the basis $\{e_1,\ldots, e_\mu\}$ and $\phi_0(e_i)=f_i$. By $R(-d_i)$ we mean a copy of the ring in which the degree of each element is shifted by $d_i$, i.e., the degree of a monomial $\prod\limits_{j=0}^{n}x_j^{a_j}$ in $R(-d_i)$ is $(\sum_{j=1}^n a_j) +d_i$. With this new grading, $\phi_0$ becomes a degree-preserving homomorphism, which means that each element of degree $d$ in $\mathbb{F}_0$  maps to an element of the same degree in $I$. The kernel of $\phi_0$ is generated by a finite number of homogeneous elements. As for $I$, the $R$-module $\ker \phi_0$ is the homomorphic image of a free $R$-module $\mathbb{F}_1$ and one may give a new grading to the elements of $\mathbb{F}_1$, so that the map $\phi_1:\ \mathbb{F}_1\rightarrow \ker \phi_0$ becomes degree-preserving.  Since $\ker \phi_0\subseteq \mathbb{F}_0$, $\phi_1$ may also be regarded as a degree-preserving map from $\mathbb{F}_1$ to $\mathbb{F}_0$, whose kernel is a homogeneous, finitely generated $R$-module. This process terminates after a finite number of steps, because for some $p\leq n$ we have $\ker \phi_{p}=0$ due to the Hilbert Syzygy Theorem. Thus we obtain an exact sequence of the form
\begin{equation}\label{eqn:mfr}
0\to  \mathbb{F}_p\to\cdots \to \mathbb{F}_2 \to \mathbb{F}_1 \to \mathbb{F}_0 \to I \to 0
\end{equation}
in which each $\mathbb{F}_i$ is a free $R$-module of the form  $\mathbb{F}_i = \oplus_j R(-j)^{\beta_{i,j}^R(I)}$. We say that the exact sequence (\ref{eqn:mfr}) has {\bf length} $p$.

\begin{defn}
An exact sequence as in (\ref{eqn:mfr}) constructed as we described above is a \textbf{graded minimal free resolution} of $I$ and the numbers $\beta_{i,j}^R(I)$ are the {\bf graded Betti numbers} of $I$.  
The {\bf Castelnuovo-Mumford regularity} of $I$, $\reg_R(I)$, is defined as $$\reg_R(I)=\max\{j-i:\ \beta_{i,j}^R(I)\neq 0\}.$$ If $\F=\{f_1,\ldots,f_m\}$ is a sequence of homogeneous polynomials, we let $\reg_R(\F)$ denote the regularity of the ideal $I=(\F)$.
The {\bf depth} of $R/I$ is the maximum length of a regular sequence in $R/I$.
The ideal $I$ is {\bf Cohen-Macaulay} if the {\bf Krull dimension} of $R/I$ is equal to its depth.
\end{defn}

It can be shown that the depth of $R/I$ is $n-p-1$, where $p$ is the length of a minimal free resolution of $I$. Although a graded minimal free resolution of a homogeneous ideal $I$ is not unique, the graded Betti numbers of $I$ are. In particular, the Castelnuovo-Mumford regularity of $I$ and the depth of $R/I$ are independent of the graded minimal free resolution used to compute them.

\begin{defn}
Let $I\subseteq R$ be a homogeneous ideal. 
The {\bf Hilbert function} of $R/I$ is the function 
$$\begin{array}{rcl}
H_{R/I}: \NN & \longrightarrow & \NN \\ d & \longmapsto & \dim_{\mathbb{K}}(R/I)_d.
\end{array}$$
The {\bf Hilbert series} of $R/I$ is the formal power series
$$HS_{R/I}(z)=\sum_{d\geq 0} H_{R/I}(d) z^d.$$
\end{defn}

\begin{notat}
Let $h(z)=\sum_{d\geq 0} h_d z^d\in\ZZ[[z]]$ be a formal power series in the variable $z$, with integer coefficients.
We denote by $\left[h(z)\right]$ the formal power series that one obtains by truncating $h(z)$ after the last consecutive positive coefficient, that is $$\left[h(z)\right]=\sum_{d=0}^\Delta h_dz^d,$$ where $\Delta=\sup\{d\geq 0\mid h_0,\ldots,h_d>0\}$.
\end{notat}

We are often interested in homogeneous ideals with the following property.

\begin{defn}
Let $I\subseteq R$ be a homogeneous ideal. 
We say that $I$ is {\bf Artinian} if there exists a $d\geq 0$ s.t. $I_d=R_d$.
\end{defn}

\begin{rmk}
The ideal $I\subseteq R$ is Artinian if and only if $HS_{R/I}(z)$ is a polynomial.
\end{rmk}

It can be shown that a homogeneous ideal $I=(f_1,\ldots,f_m)\subseteq R$ is Artinian if and only if $\F=\{f_1,\ldots,f_m\}$ contains a regular sequence of $n$ polynomials. This is the case if and only if the system $f_1=\cdots=f_m=0$ has no projective solutions, i.e., the only solution of the system is $x_1=\cdots=x_n=0$.

Throughout the paper, we often consider inhomogeneous systems $\F$ with at least a solution. For such a system, $J=(\F^h)$ is not Artinian. However, if $J$ is Cohen-Macaulay, one can consider an Artinian reduction. We give the definition of Artinian reduction only in the special case which will interest us.

\begin{defn}\label{defn:artinred}
Let $\F^h\subseteq S$ be a homogeneous system of polynomial equations which has finitely many (projective) solutions over the algebraic closure of $\mathbb{K}$.
Let $\ell\in S_1$ be a homogeneous linear form such that $\ell\nmid 0$ modulo $J=(\F^h)$. The ideal $$H=J+(\ell)/(\ell)\subseteq S/(\ell)$$ is an {\bf Artinian reduction} of $J$.
\end{defn}

Notice that a linear form $\ell$ as in Definition~\ref{defn:artinred} may not exist. In that case, the ideal $J$ does not have an Artinian reduction over $\mathbb{K}$.

\begin{rmk}\label{rmk:artinred}
In the situation of Definition~\ref{defn:artinred} one has
$$\reg_{S}(J)=\reg_{S/(\ell)}(H).$$
\end{rmk}

\subsection{Homogeneous semi-regular sequences}

In cryptography, we are often interested in analyzing the behavior of sequences of polynomials which are chosen ``at random''. In algebraic geometry, this can be formalized via the concept of genericity. A property is {\bf generic} or {\bf holds generically} if there exists a nonempty Zariski-open set where the property holds. By identifying a polynomial with the vector of its coefficients, both the set of homogeneous polynomials of degree $d$ and that of arbitrary polynomials of degree $\leq d$ can be regarded as a projective space. Notice that we identify polynomials which are the same up to a nonzero scalar multiple. Hence a {\bf generic homogeneous polynomial of degree $d$} is a homogeneous polynomial of degree $d$, which belongs to a given nonempty Zariski-open set in the projective space of all homogeneous polynomials of degree $d$. Similarly, a {\bf generic polynomial of degree $\leq d$} is a polynomial of degree $\leq d$, which belongs to a given nonempty Zariski-open set in the projective space of all polynomials of degree $\leq d$. Along the same lines, one defines a {\bf generic sequence} of polynomials. 

Notice that, in order for the concept of genericity to be meaningful, one needs to work over an {\bf infinite field}. In fact, over an infinite field, a nonempty Zariski-open set is dense, i.e., its closure is the whole space. Over a finite field, on the other side, the Zariski topology is the discrete topology. Hence, over a finite field, every set of polynomials is a Zariski-open set. In particular, a proper Zariski-open set is never dense over a finite field. Therefore, over a finite field, a generic property is no longer a property which is true ``almost everywhere".

Semi-regular sequences were first introduced by Pardue in~\cite{P99}, which was later expanded and published as~\cite{P10}. 

\begin{defn}\label{defn:srs}
Let $R=\mathbb{K}[x_1,\ldots,x_n]$ and assume that $\mathbb{K}$ is an infinite field. 
If $A=R/I$, where $I$ is a homogeneous ideal, and $f\in R_d$, then $f$ is {\bf semi-regular} on $A$ if for every $e\geq d$, the vector space map $A_{e-d}\rightarrow A_e$ given by multiplication by $f$ is of maximal
rank (that is, either injective or surjective). A sequence of homogeneous polynomials $f_1,\ldots,f_m$ is a {\bf semi-regular sequence} if each $f_i$ is semiregular on $A/(f_1,\ldots,f_{i-1})$, $1\leq i\leq m$.
\end{defn}

\begin{rmk}
If $m\leq n$, then $f_1,\ldots,f_m$ is a semi-regular sequence if and only if it is a regular sequence. 
\end{rmk}

Semi-regular sequences are conjectured by Pardue to be generic sequences of polynomials in~\cite[Conjecture~B]{P10}. He proves that this fact is equivalent to Fr\"oberg's Conjecture~\cite{Fr}, which is known to be true for many classes of ideals. See \cite{T19} for an up-to-date list of these classes.

Therefore, assuming that \cite[Conjecture~B]{P10} holds, ``random'' sequences of polynomials with coefficients in an infinite field $\mathbb{K}$ are semi-regular. One advantage of dealing with semi-regular sequences is that their Hilbert function is known.

\begin{prop}[{{\cite[Proposition~1]{P10}}}]
	\label{prop:hssemireg}
	Let $f_1,\ldots,f_m\in R$ be homogeneous polynomials of degrees $d_1,\ldots,d_m$. Then $f_1,\ldots,f_m$ is a semi-regular sequence on $R$ if and only if $$HS_{R/(f_1,\ldots,f_\ell)}(z)=\left[\frac{\prod_{i=1}^\ell (1-z^{d_i})}{(1-z)^n}\right]$$ for $1\leq \ell\leq m$.
\end{prop}

The interest of semi-regular sequences for multivariate cryptography was first observed in~\cite{bardet2004complexity}. Since the definition of semi-regular sequences adopted by Bardet, Faug{\`e}re, and Salvy differs by the one given by Pardue, we will use the term {\bf cryptographic semi-regular sequence}.  The definition we give below is equivalent to~\cite[Definition~3]{bardet2004complexity}, as shown in~\cite[Proposition~3.2.5]{B04}.

\begin{defn}\label{defn:csrs}
Let $\mathbb{K}$ be an arbitrary field.
A sequence of homogeneous polynomials $f_1,\ldots,f_m\in R$ is a {\bf cryptographic semi-regular sequence} if $$HS_{R/(f_1,\ldots,f_m)}(z)=\left[\frac{\prod_{i=1}^m (1-z^{d_i})}{(1-z)^n}\right].$$
\end{defn}

\begin{rmk}
Any cryptographic semi-regular sequence with $m\geq n$ generates an Artinian ideal.
\end{rmk}

Notice that Definition~\ref{defn:srs} makes sense also over an arbitrary field. We now briefly compare semi-regular sequences over an arbitrary field with cryptographic semi-regular sequences.
First of all, any semi-regular sequence is a cryptographic semi-regular sequence by Proposition~\ref{prop:hssemireg}. The fact that the converse does not hold follows from the example just above~\cite[Conjecture~B]{P10}. 

Above we discussed a conjecture by Pardue which implies that sequences of polynomials whose coefficients are chosen uniformly at random over an infinite field are semi-regular with high probability. Pardue's Conjecture and the related Fr\"oberg's Conjecture have been extensively studied and there is evidence in support of their correctness. In~\cite[Conjecture~2]{BFS}, Bardet, Faug\`ere, and Salvy conjecture that, for $n\rightarrow\infty$, the proportion of cryptographic semi-regular sequences in the set of all sequences of $m$ polynomials in $\mathbb{F}_2[x_1,\ldots,x_n]$ of degrees $d_1,\ldots,d_m$ tends to 1.  In~\cite[Theorem~7.14]{HMS17}, Hodges, Molina, and Schlather disprove this conjecture and prove that the proportion tends to 0 as $n$ tends to $\infty$. They also propose variations of the conjecture that ``most" sequences of $m$ polynomials in $\mathbb{F}_2[x_1,\ldots,x_n]$ are semi-regular and prove some related results. To the extent of our knowledge, none of these conjectures has been studied for finite fields different from $\mathbb{F}_2$.

\subsection{The Macaulay matrix and the solving degree of a system of equations}

For $d\geq 1$, the {\bf Macaulay matrix} $M_d(\F)$ of a polynomial system $\F=\{f_1,\ldots,f_m\}$ is a matrix with entries in $\mathbb{K}$ whose columns are indexed by all elements of $\mathrm{Mon}(R)$ of degree $\leq d$, sorted in decreasing order from left to right with respect to the degree reverse lexicographic order. The rows are indexed by the polynomials $m_if_j$, where $f_j\in\F$, $m_i\in\mathrm{Mon}(R)$, and $\mathrm{deg}\ m_if_j\leq d$. The $(k,l)$-th entry of the matrix is the coefficient of the index of column $l$ in the polynomial which is the index of row $k$. 

In order to compute a Gr\"obner basis of $I$, one performs Gaussian elimination on the Macaulay matrix for increasing values of $d$. The complexity of computing the reduced row echelon form (RREF) of these matrices is bounded by a known function of the solving degree, which is the largest degree which is involved in the computation.
Therefore, the solving degree is the relevant parameter to estimate, in order to estimate the complexity of computing the solutions of the system $\F$. 

We analyze the following algorithm to compute the reduced Gr\"obner basis of $I$ with respect to the degree reverse lexicographic order. Start in degree $d=\max\{d_1,\ldots,d_m\}$. Perform Gaussian elimination on $M_d(\F)$ to compute its RREF. Since the rows of $M_d(\F)$ correspond to the polynomials $m_if_j$, Gaussian elimination corresponds to taking linear combinations of these polynomials. Hence, every row in the RREF corresponds to a polynomial in the ideal generated by $\mathcal{F}$. In order to better keep track of what happens to each row, we use a variant of Gaussian elimination which does not permute the rows. Suppose the $k$th row of $M_d(\F)$ is indexed by the polynomial $m_if_j$. Then the $k$th row in the RREF corresponds to a polynomial of the form [$m_if_j$ $+$ a linear combination of other rows of $M_d(\F)$]. If computing the RREF produces a polynomial $f$ which has leading term strictly smaller than that of $m_if_j$ and $\deg(f)<d$, then one appends to the matrix new rows $uf$ for all $u\in\mathrm{Mon}(R)$ such that $\deg(uf)\leq d$. This condition is checked for all rows. Gaussian elimination is then performed on the resulting matrix, and the process is repeated. Eventually, no degree reductions will be produced. Then we have either found a Gr\"obner basis of $I$ and we stop, or we have not and we proceed to the next degree, $d+1$.

The algorithm as described will compute a Gr\"obner basis for $I$. It does not, however, give a method for verifying whether the final matrix output corresponds to a Gr\"obner basis. One stopping criterion is that the S-polynomials corresponding to the output basis reduce to 0. Suppose we want to verify the output after $d$ iterations. The stopping criterion can be verified by Gaussian elimination, however, this will be of a matrix in degree $d’$, where $d <d’\le 2d-1$. Another possible stopping criterion is giving an a priori bound on the solving degree. Concretely, if one can prove that the solving degree of a system $\mathcal{F}$ is at most $D$, then one can stop the computation in degree $D$.

\begin{defn} Suppose $n \le m$ and $d_1\le\cdots\le d_m$. The {\bf Macaulay bound} is $$\sum\limits_{i=m-n+1}^m(d_i-1)+1.$$
\end{defn}

The Macaulay bound was shown by Lazard \cite{lazard1983grobner} to bound from above the degrees of the polynomials in a Gr\"obner basis of $(\mathcal{F})$, for a homogeneous system $\mathcal{F}$ that has finitely many solutions over the algebraic closure of $\mathbb{K}$.

\begin{defn}
The {\bf solving degree} of $\F$, $\sd(\F)$, is the least degree $d$ in which the algorithm described above returns a degree reverse lexicographic Gr\"obner basis of $I$.
\end{defn}

Intuitively, the solving degree is the largest degree of the polynomials involved in the computation of the reduced degree reverse lexicographic Gr\"obner basis of $I$.

\section{Solving degree,  degree of regularity, and Castelnuovo-Mumford regularity}\label{sect:dreg}

Let $\F=\{f_1,\ldots,f_m\}\subseteq R$ and consider the degree reverse lexicographic order on $R$.
We are interested in bounding the solving degree of $\F$.

The concept of degree of regularity is widely used in the cryptographic literature.
\begin{defn}[{{\cite[Definition~4]{bardet2004complexity}}}]
	Let $\F$ be a system of polynomial equations and assume that $(\F^{\ttop})_d=R_d$ for $d\gg 0$.
	The {\bf degree of regularity} of $\F$ is $$d_{\reg}(\F)=\min\{d\geq 0\mid (\F^{\ttop})_d=R_d\}.$$
	If $(\F^{\ttop})_d\neq R_d$ for all $d\geq 0$, we let $d_{\reg}(\F)=\infty$.
\end{defn}

It follows from~\cite[Proposition~4.5]{CG17} that if $I_d^{\ttop}= R_d$ for some $d\geq 0$, then 
$$d_{reg}(\mathcal{F})=\reg(\F^{\ttop}).$$

Many authors use the degree of regularity as a heuristic upper bound or estimate for the solving degree. 
To the best of our knowledge, however, this has never been formalized and we could not find a proof in the literature that the degree of regularity produces an estimate for the solving degree. The next examples show that the degree of regularity is not an upper bound for the solving degree and that their difference can be non-negligible.

The next two examples are inspired by an example in \cite[Page 10]{Beng}. In both of them we use the largest {\bf step degree} computed via the computer algebra system Magma~\cite{magma}, as a proxy for the solving degree. More precisely, we use the value of the largest step degree which appears in the Magma computation as a value of the solving degree. 

\begin{ex}
Let $R=\FF_7[x,y,z]$ and let $f_x=x^7-x$, $f_y=y^7-y$, $f_z=z^7-z$ be the field equations. Consider the equations 
$$f_1=x^5+y^5+z^5-1,\; f_2=x^3+y^3+z^2-1,\; f_3=y^6-1,\; f_4=z^6-1.$$
Consider the systems of equations
$$\F=\left\{\left. \prod_{j=1}^3 f_{i_j} \right| 1\leq i_1\leq i_2\leq i_3\leq 4\right\}\cup\{f_x,f_y,f_z\}.$$
Observe that the system generates a radical ideal, since it contains the field equations.
Using Magma one can compute
$$\sd(\F)=24>15=d_{\reg}(\F).$$
Notice that $\F$ contains equations of degree 18, however, this still does not account for the gap between the solving degree and the degree of regularity of $\F$.
\end{ex}

\begin{ex}
Let $R=\FF_7[x,y,z]$ and let $f_x=x^7-x$, $f_y=y^7-y$, $f_z=z^7-z$ be the field equations. Consider the equations 
$$f_1=x^5+y^5+z^5-1,\; f_2=x^3+y^3+z^2-1,\; f_3=f_x,\; f_4=f_y,\; f_5=f_z.$$
Consider the systems of equations
$$\F=\left\{\prod_{1\leq i\leq j\leq 5} f_i f_j \right\}\cup\{f_x,f_y,f_z\}.$$
Observe that $\F$ generates a radical ideal, since it contains the field equations.
Using Magma one can compute
$$\sd(\F)=21>13=d_{\reg}(\F).$$
Notice that $\F$ contains equations of degree 14, however, this still does not account for the gap between the solving degree and the degree of regularity of $\F$.
\end{ex}

\begin{ex}\label{gap1}
Let $\mathcal{F}=\{x^4-1,x^2y-x^2,y^2-1\}\subseteq \mathbb{Z}_7[x,y]$. The ideal $(\F^{\ttop})=(x^4,x^2y,y^2)$ is generated by a cryptographic semi-regular sequence and $d_{\reg}(\F)=4$. The reduced Gr\"obner bases of the ideal $I=(\mathcal{F})$ with respect to the degree reverse lexicographic order with $x>y$ is $\{y-1,x^4-1\}$ and $\sd(\mathcal{F})=5>4=d_{\reg}(\F).$
%
%
\end{ex}

The regularity of the ideal generated by the system $\F^h$ obtained by homogenizing the equations of $\F$ with respect to a new variable is an upper bound on the solving degree of $\F$, whenever the equations of $\F^h$ generate an ideal in generic coordinates. This is the case in particular if $\F$ contains the field equations. We refer the reader to~\cite[Definition~1.11]{CG17} for the definition of generic coordinates.

\begin{thm}[{{\cite[Theorem~3.23 and Theorem~3.26]{CG17}}}]\label{CGbound}
Consider the degree reverse lexicographic order on $R$.
Assume that $\F=\{f_1,\ldots,f_m\}\subseteq R$ contains the field equations and let $J=(f_1^h,\ldots,f_m^h)\subseteq S$. 
Then 
$$\sd(\mathcal{F})\leq \reg(J).$$
\end{thm}

It follows from Theorem~\ref{CGbound} that, whenever $t\nmid 0$ modulo $J=(\F^h)$, then the degree of regularity essentially bounds the solving degree. Unfortunately, it is often the case that $t\mid 0$ modulo $J$, as we discuss in Section~\ref{sect:eghdreg}. Therefore, the applicability of the next result is limited.

\begin{cor}
Consider the degree reverse lexicographic order on $R$. Assume that $\F=\{f_1,\ldots,f_m\}\subseteq R$ contains the field equations and let $J=(f_1^h,\ldots,f_m^h)\subseteq S$. If $t\nmid 0$ modulo $J$, then
$$\sd(\mathcal{F})\leq d_{\reg}(\F).$$
\end{cor}

\begin{proof}
Let $(\F^{\ttop})=(f_1^{\ttop},\ldots,f_m^{\ttop})\subseteq R$. The natural isomorphism $S/(t)\cong R$ maps $J+(t)/(t)$ to $(\F^{\ttop})$. If $t\nmid 0$ modulo $J$, then $(\F^{\ttop})$ is an Artinian reduction of $J$, hence 
$$\reg(J)=\reg(\F^{\ttop})=d_{\reg}(\F).$$
The thesis now follows from Theorem~\ref{CGbound}.
\end{proof}

\section{Solving degree of cryptographic semi-regular systems}\label{semi-regular}

Let $\F=\{f_1,\ldots,f_m\}\subseteq R$ be a system of homogeneous or inhomogeneous equations. In the first  subsection we provide explicit bounds on the solving degree of $\F$ when the polynomials are homogeneous. In the second subsection we propose a definition of semi-regularity for inhomogeneous polynomials, which is different from that of~\cite{BFSY}. Then we derive explicit bounds on the solving degree of $\F$ when the polynomials are inhomogeneous. In addition, in the appendix we provide a table of upper bounds for the solving degree of cryptographic semi-regular systems with $n$ variables and $m=n+k$ quadratic equations, for $1 \le n,k \le 100$.

\subsection{Homogeneous cryptographic semi-regular sequences}

The definition of cryptographic semi-regular sequence implicitly provides a bound for the solving degree of the corresponding system. In fact, since the Hilbert series of a cryptographic semi-regular sequence is given by the formula in Definition~\ref{defn:csrs}, for any given choice of $m$, $n$ and the degrees $d_1,\ldots,d_m$, one can use the formula to compute the Castelnuovo-Mumford regularity of the ideal $I$ generated by the equations of the system $\mathcal{F}$. Explicitly, the degree of the term with the first non-positive coefficient in the Hilbert series expansion is equal to the Castelnuovo-Mumford regularity of $I$. The Castelnuovo-Mumford regularity of $I$ provides a bound for the solving degree of $\mathcal{F}$ by~\cite[Theorem~3.22]{CG17}, under the assumption that the ideal $I$ is in generic coordinates. However, even if such a bound can be computed for any choice of $m,n,d_1,\ldots,d_m$, it is a difficult task to provide an explicit formula for it.

If $m\leq n$, then a cryptographic semi-regular sequence is a regular sequence. The Castelnuovo-Mumford regularity of a regular sequence is given by a simple formula, which coincides with the Macaulay bound for $m=n$. If $m>n$, asymptotic formulas for the degree of regularity of a cryptographic semi-regular sequence (which in this situation coincides with the Castelnuovo-Mumford regularity of the ideal that the sequence generates) are given in~\cite{bardet2004complexity,BFSY}.

In this subsection we produce new explicit bounds for the solving degree of homogeneous cryptographic semi-regular systems. Instead of making an asymptotic analysis, we study the situation when the difference $m-n$ is small. More precisely, we provide explicit formulas in the cases: $m=n+1$, $n+2\leq m\leq n+5$ and the equations are quadratic, or $m=n+2$ and the equations are cubic.

Our first result concerns the solving degree of systems consisting of $n+1$ generic homogeneous polynomials $f_1,\ldots,f_{n+1}$ of degrees $d_i=\deg(f_i)$, $d_1\leq\cdots\leq d_{n+1}$ over an infinite field. This can be easily computed by using a result of Migliore and Mir\`o-Roig~\cite{MM}. Notice that, since the polynomials are generic and the field is infinite, then $f_1,\ldots,f_n$ are a complete intersection. Therefore, $(f_1,\ldots,f_n)_d=R_d$ for $d\geq d_1+\cdots+d_n-n+1$. Therefore, if $d_{n+1}\geq d_1+\cdots+d_n-n+1$, then the last equation may be removed from the system, as $f_{n+1}\in(f_1,\ldots,f_n)$. Hence we may assume without loss of generality that $d_{n+1}\leq d_1+\cdots+d_n-n$.

\begin{thm}\label{generic_aci}
Let $\mathbb{K}$ be an infinite field and let $\F=\{f_1,\ldots,f_{n+1}\}$ consist of $n+1$ generic homogeneous polynomials of degrees $d_i=\deg(f_i)$ in $n$ variables. Let $d_1\leq d_2\leq\cdots\leq d_{n+1}$. Assume without loss of generality that $d_{n+1}\leq d_1+\cdots+d_n-n$.
Then $\F$ is a cryptographic semi-regular sequence and 
$$\sd(\F)\leq\left\lfloor\frac{d_1+\cdots+d_{n+1}-n-1}{2}\right\rfloor +1.$$
In particular, if $d_1=\cdots=d_{n+1}=2$, then 
$$\sd(\F)\leq\left\lfloor\frac{n+1}{2}\right\rfloor +1$$
and if $d_1=\cdots=d_{n+1}=3$, then 
$$\sd(\F)\leq n+2.$$
\end{thm}

\begin{proof}
Let $I=(f_1,\ldots,f_{n+1})\subseteq R$. The fact that $\F$ is a cryptographic semi-regular sequence is explained in~\cite{MM}, using results from~\cite{S80,W87}. By \cite[Lemma~2.5]{MM} one has 
$$\reg(I)=\left\lfloor\frac{d_1+\cdots+d_{n+1}-n-1}{2}\right\rfloor +1.$$
If $d_1=\cdots=d_{n+1}=2$, then $$\reg(I)=\left\lfloor\frac{n+1}{2}\right\rfloor +1.$$
If $d_1=\cdots=d_{n+1}=3$, then $$\reg(I)\leq n+2.$$
Since $I$ is generated by generic polynomials, then it is in generic coordinates. Therefore we conclude by~\cite[Theorem~3.22]{CG17}.
\end{proof}

The same result holds for homogeneous cryptographic semi-regular systems over a finite field. Over a field of characteristic zero, a random sequence of $n+1$ polynomials will be semi-regular, hence also cryptographic semi-regular, see~\cite{PR09}. However, we cannot make the same claim for polynomials over a finite field. Hence, in Theorem \ref{generic_aci_ff}, we specify $\mathcal{F}$ to be a homogeneous cryptographic semi-regular sequence. The proof of Theorem \ref{generic_aci_ff} is the same as the proof of Theorem \ref{generic_aci}.

\begin{thm}\label{generic_aci_ff}
Let $\mathbb{K}$ be a finite field and let $\F=\{f_1,\ldots,f_{n+1}\}$ be a homogeneous cryptographic semi-regular sequence of  polynomials of degrees $d_i=\deg(f_i)$ in $n$ variables. Let $I=(f_1,\ldots,f_{n+1})$. Let $d_1\leq d_2\leq\cdots\leq d_{n+1}$. Assume without loss of generality that $d_{n+1}\leq d_1+\cdots+d_n-n$.
Then 
$$\maxgb(I)\leq\left\lfloor\frac{d_1+\cdots+d_{n+1}-n-1}{2}\right\rfloor +1.$$
In particular, if $d_1=\cdots=d_{n+1}=2$, then 
$$\maxgb(I)\leq\left\lfloor\frac{n+1}{2}\right\rfloor +1$$
and if $d_1=\cdots=d_{n+1}=3$, then 
$$\maxgb(I)\leq n+2.$$

If in addition $I$ is in generic coordinates, then
$$\sd(\F)\leq\left\lfloor\frac{d_1+\cdots+d_{n+1}-n-1}{2}\right\rfloor +1.$$
In particular, if $d_1=\cdots=d_{n+1}=2$ and $I$ is in generic coordinates, then 
$$\sd(\F)\leq\left\lfloor\frac{n+1}{2}\right\rfloor +1$$
and if $d_1=\cdots=d_{n+1}=3$ and $I$ is in generic coordinates, then 
$$\sd(\F)\leq n+2.$$
\end{thm}

\begin{rmk}
Notice that the bound obtained in Theorem~\ref{generic_aci} and Theorem~\ref{generic_aci_ff} implies that $$\sd(\F)\leq d_1+\ldots+d_n-n,$$ since $d_{n+1}\leq d_1+\ldots+d_n-n$. In particular, this bound is always better than the Macaulay bound.
\end{rmk}

We now study the case when $n+2\leq m\leq n+5$ and the equations are homogeneous and quadratic. The assumption that $\F$ generates an ideal which is in generic coordinates is satisfied for ``sufficiently general'' polynomials, or when $\F$ contains the homogenizations of the field equations.

\begin{thm}\label{d-reg}
Let $\F=\{f_1,\ldots,f_m\}$ be a cryptographic semi-regular sequence of homogeneous polynomials of degree $2$ in $n$ variables. Let $I=(\F)$ and let
\begin{align*}
r(m,n)=\begin{cases}
\left\lceil(4+n-\sqrt{4+n})/2\right\rceil\quad &\text{if } m=n+2,\vspace{.2cm}\\\vspace{.2cm}
\left\lceil (6 + n - \sqrt{16 + 3 n})/2\right\rceil\quad &\text{if } m=n+3,\\\vspace{.2cm}
\left\lceil(8 + n - \sqrt{20 + 3 n + \sqrt{2} \sqrt{128 + 39 n + 3 n^2}})/2\right\rceil\quad &\text{if } m=n+4,\\
\left\lceil(10 + n - \sqrt{40 + 5 n + \sqrt{2}\sqrt{288 + 75 n + 5 n^2}})/2\right\rceil\quad &\text{if } m=n+5.
\end{cases}
\end{align*}
Then $$\maxgb(I)\leq r(m,n).$$
If in addition we assume that $I$ is in generic coordinates, then $$\sd(\F)\leq r(m,n).$$
\end{thm}

\begin{proof}
Let $I=(f_1,\ldots,f_m)\subseteq R$.
Since $m>n$ and $\F=\{f_1,\ldots,f_m\}$ is a cryptographic semi-regular sequence, then there exists a $d$ such that $I_d=R_d$. The Castelnuovo-Mumford regularity of $I$ is the least such degree. By the definition of cryptographic semi-regular sequence, $\reg(I)$ is the least degree in which the formal power series $(1-z^{2})^m/(1-z)^n$ has a non-positive coefficient.
One has 
\begin{align}
\frac{(1-z^{2})^m}{(1-z)^n}&=(1-z)^{m-n}(1+z)^{m}=\left(\sum_{j=0}^{m-n}(-1)^{j}\binom{m-n}{j}z^{j}\right)\left(\sum_{i=0}^{m}\binom{m}{i}z^{i}\right)\nonumber\\
&=\sum_{j=0}^{m-n}\left(\sum_{i=0}^{m}(-1)^j\binom{m-n}{j}\binom{m}{i}z^{i+j}\right)\label{Hilbert series}.
\end{align}
Hence the coefficient of $z^k$ in (\ref{Hilbert series}) is 
\begin{align*}
C_k=\sum_{l=0}^k \left((-1)^l\binom{m-n}{l}\binom{m}{k-l}\right)
\end{align*}
and for $k\geq m-n$ we have
\begin{align*}
C_k&=\sum_{l=0}^{m-n} \left((-1)^l\binom{m-n}{l}\binom{m}{k-l}\right)\\
&=\frac{m!(m-n)!}{k!(2m-n-k)!}\sum_{l=0}^{m-n} \left((-1)^l\binom{2m-n-k}{m-n-l}\binom{k}{l}\right)
\end{align*}
Setting $r:=m-n$, we have 
\begin{align}\label{sign}
C_k=\frac{(n+r)!}{k!(2r+n-k)!}f(r,k),
\end{align}
where
\begin{align}\label{f(r,k)}
f(r,k)=r!\sum_{l=0}^{r} \left((-1)^l\binom{2r+n-k}{r-l}\binom{k}{l}\right).
\end{align}	
By (\ref{sign}), $C_k$ and $f(r,k)$ have the same sign. Hence in order to find $\reg(I)$, it suffices to find the smallest $k$ for which $f(r,k)$ is non-positive.
	
(i) Letting $m=n+2$ in (\ref{f(r,k)}), we obtain
\begin{align*}
f(2,k)=4k^2-4(4+n)k+n^2+7n+12.
\end{align*}
As a function of $k$, $f(2,k)$ has two zeros
$$k_1=(4+n-\sqrt{4+n})/2, \quad k_2=(4+n+\sqrt{4+n})/2.$$
So it  is positive in $(-\infty, k_1)\cup(k_2,+\infty)$, and is negative in $(k_1,k_2)$.
Since $\lceil k_1\rceil<k_2$, the first non-positive coefficient of $z^k$ in (\ref{Hilbert series}) occurs for $k=\lceil (4+n-\sqrt{4+n})/2\rceil$. 

\medskip
(ii) Letting $m=n+3$ in (\ref{f(r,k)}), we obtain
\begin{align*}
f(3,k)=- 8 k^3+ 12  (6 + n)k^2- 2 (92 + 33 n + 3 n^2)k+n^3+15n^2+74n+120.
\end{align*}
As a function of $k$,  $f(3,k)$ has  three zeros
$$k_1=(6 + n - \sqrt{16 + 3 n})/2, \quad k_2=(6 + n)/2, \quad k_3=(6 + n + \sqrt{16 + 3 n})/2,$$
 it is positive in $(-\infty, k_1)\cup(k_2,k_3)$, and negative in $(k_1,k_2)\cup(k_3,+\infty)$. Since $\lceil k_1\rceil<k_2$, then $\reg(I)=\lceil k_1\rceil.$
	
\medskip
(iii) Letting $m=n+4$ in (\ref{f(r,k)}), we obtain
\begin{align*}
f(4,k)= & 16 k^4 - 32(8 + n)k^3+ 8(172 + 45 n + 3 n^2)k^2- 8(352 + 148 n + 21 n^2 + n^3)k\\
&  + n^4+ 26 n^3+  251 n^2 +1066 n +1680 .
\end{align*}
As a function of $k$, $f(4,k)$ admits four zeros 
\begin{eqnarray*}
k_1&=&(8 + n - \sqrt{20 + 3 n + \sqrt{2} \sqrt{128 + 39 n + 3 n^2}})/2,\\
k_2&=&(8 + n - \sqrt{20 + 3 n - \sqrt{2}\sqrt{128 + 39 n + 3 n^2}})/2,\\
k_3&=&(8 + n + \sqrt{20 + 3 n - \sqrt{2}\sqrt{128 + 39 n + 3 n^2}})/2,\\
k_4&=&(8 + n + \sqrt{20 + 3 n + \sqrt{2}\sqrt{128 + 39 n + 3 n^2}})/2.
\end{eqnarray*}
This function is positive in $(-\infty,k_1)\cup(k_2,k_3)\cup(k_4,+\infty)$, and negative in $(k_1,k_2)\cup(k_3,k_4)$. Since $\lceil k_1\rceil<k_2$, then $\reg(I)=\lceil k_1\rceil$.
	
\medskip
(iv) Letting $m=n+5$ in (\ref{f(r,k)}), we obtain
\begin{align*}
f(5,k)=& - 32 k^5 + 80 (10 + n)k^4- 80  (92 + 19 n + n^2)k^3\\
&	- 2  (27024 + 12450 n + 2175 n^2 + 170 n^3 + 5 n^4)k		+40  (760 + 246 n + 27 n^2 + n^3)k^2\\
&+ n^5 + 40 n^4 +  635 n^3 + 5000 n^2 + 19524 n   +30240.
\end{align*}
As a function of $k$, $f(5,k)$ admits five zeros
\begin{eqnarray*}
i_1&=&\frac{1}{2} \left(10 + n - \sqrt{40 + 5 n + \sqrt{2} \sqrt{288 + 75 n + 5 n^2}}\right)\\
i_2&=&\frac{1}{2} \left(10 + n - \sqrt{40 + 5 n - \sqrt{2} \sqrt{288 + 75 n + 5 n^2}}\right)\\
i_3&=&\frac{10+n}{2}\\
i_4&=&\frac{1}{2} \left(10 + n +\sqrt{40 + 5 n -\sqrt{2} \sqrt{288 + 75 n + 5 n^2}}\right)\\
i_5&=&\frac{1}{2} \left(10 + n + \sqrt{40 + 5 n + \sqrt{2} \sqrt{288 + 75 n + 5 n^2}}\right)
\end{eqnarray*}
Since the polynomial function $f(k,5)$ is continuous and postive in $i\in(-\infty,i_1)\cup(i_2,i_3)\cup(i_4,i_5)$, the first change in sign from positive to negative occurs when 
$$k=\left\lceil\frac{1}{2} \left(10 + n - \sqrt{40 + 5 n + \sqrt{2} \sqrt{288 + 75 n + 5 n^2}}\right)\right\rceil.$$

We have proved that the Castelnuovo-Mumford regularity of $I$ is $\reg(I)=r(m,n)$. The bound on the degree of the elements of the degree reverse lexicographic Gr\"obner basis of $I$ now follows from~\cite[Proposition~4.5 and Remark~4.6]{CG17}. The bound on the solving degree follows from~\cite[Theorem~3.22]{CG17}, under the assumption that $I$ is in generic coordinates.
\end{proof}

Theorem~\ref{generic_aci_ff} and Theorem~\ref{d-reg} also provide an upper bound for the solving degree of homogeneous cryptographic semi-regular sequences for larger values of $m$.

\begin{cor}\label{cor:largerm}
Let $\F=\{f_1,\ldots,f_m\}$ be a cryptographic semi-regular sequence of homogeneous polynomials of degree $d=2,3$ in $n$ variables. Assume that $m\geq n+5$ if $d=2$ and that $m\geq n+1$ if $d=3$. 
Let $I=(\F)$ and let 
$$r(n,d)=\left\{\begin{array}{ll}
\left\lceil(10 + n - \sqrt{40 + 5 n + \sqrt{2}\sqrt{288 + 75 n + 5 n^2}})/2\right\rceil & \mbox{ if } d=2,\\
n+2 & \mbox{ if } d=3.
\end{array}\right.$$
Then $$\maxgb(I)\leq r(n,d).$$
If in addition we assume that $I$ is in generic coordinates, then
$$\sd(\F)\leq r(n,d).$$
\end{cor}

\begin{proof}
Notice that $r(n,2)$ is the Castelnuovo-Mumford regularity of an ideal $H$ generated by a cryptographic semi-regular sequence consisting of $n+5$ homogeneous quadratic polynomials. Moreover, $r(n,3)$ is the Castelnuovo-Mumford regularity of an ideal $H$ generated by a cryptographic semi-regular sequence consisting of $n+1$ homogeneous cubic polynomials. Since $I$ contains such an ideal $H$ and $H$ is Artinian, then $\reg(I)\leq \reg(H)$.

The bound on the degree of the elements of the degree reverse lexicographic Gr\"obner basis of $I$ now follows from~\cite[Proposition~4.5 and Remark~4.6]{CG17}. The bound on the solving degree of $\F$ follows from~\cite[Theorem~3.22]{CG17}, under the assumption that $I$ is in generic coordinates.
\end{proof}

\subsection{Inhomogeneous cryptographic semi-regular sequences}

Let $\F=\{f_1,\ldots,f_m\}\subseteq R$ be a system of inhomogeneous polynomial.
Inhomogeneous cryptographic semi-regular sequences are defined in~\cite[Definition~5]{BFSY} as sequences $\F$ such that $\F^{\ttop}$ is a cryptographic semi-regular sequence, according to Definition~\ref{defn:csrs}. This definition allows the authors to estimate the degree of regularity of $\F$. 

The examples that we presented in Section~\ref{sect:dreg}, however, show that the degree of regularity of $\F$ can be quite a bit smaller than its solving degree. In view of those examples, therefore, we propose a different definition. Instead of looking at the sequence $\F^{\ttop}$ of homogeneous parts of highest degree, we consider the sequence $\F^h$ of the homogenizations of the original polynomials, with respect to a new variable.

\begin{defn}\label{defn:icrsr}
An  inhomogeneous system of polynomials $\F=\{f_1,\ldots,f_m\}\subseteq R$ is a {\bf cryptographic semi-regular sequence} if $\F^h=\{f_1^h,\ldots,f_m^h\}\subseteq S=R[t]$ is a cryptographic semi-regular sequence.
\end{defn}

Definition~\ref{defn:icrsr} is very natural also in view of Pardue's Conjecture~\cite[Conjecture~B]{P10}. Informally, Pardue's Conjecture states that semi-regular sequences are sequences of generic polynomials, i.e. random systems of  polynomials. If we think of a random inhomogeneous polynomial $f\in R$ of degree $d$ as a linear combination with randomly chosen coefficients of the monomials of $R$ of degree less than or equal to $d$, then its homogenization $f^h\in S$ is a linear combination with randomly chosen coefficients of the monomials of $S$ of degree $d$. In other words, $f$ is a random inhomogeneous polynomial of degree $d$ if and only if $f^h$ is a random homogeneous polynomial of degree $d$.

Definition~\ref{defn:icrsr} allows us to apply our results on homogeneous systems from the previous subsection to systems of inhomogeneous polynomials. As a direct consequence of Theorem~\ref{generic_aci}, Theorem~\ref{generic_aci_ff}, Theorem~\ref{d-reg}, and Corollary~\ref{cor:largerm}, we obtain the following results.

\begin{thm}\label{thm:inhomog}
Let $\mathbb{K}$ be an infinite field and let $\F=\{f_1,\ldots,f_m\}\subseteq R$ be a sequence of generic inhomogeneous polynomials of degrees $d_i=\deg(f_i)$, with $m\in\{n+1,n+2\}$. If $m=n+2$, assume without loss of generality that $d_{n+2}\leq d_1+\cdots+d_{n+1}-n-1$. Then $\F$ is a cryptographic semi-regular sequence and 
$$\sd(\F)\leq\left\{\begin{array}{ll}
d_1+\cdots+d_{n+1}-n & \mbox{ if } m=n+1,\\
\left\lfloor\frac{d_1+\cdots+d_{n+2}-n-2}{2}\right\rfloor+1 & \mbox{ if } m=n+2.
\end{array}\right.$$
In particular, if $d_1=\cdots=d_m=2$, then 
$$\sd(\F)\leq\left\{\begin{array}{ll}
n+2 & \mbox{ if } m=n+1,\\
\left\lfloor\frac{n}{2}\right\rfloor +2 & \mbox{ if } m=n+2,
\end{array}\right.$$
and if $d_1=\cdots=d_{n+1}=3$, then 
$$\sd(\F)\leq\left\{\begin{array}{ll}
2n+3 & \mbox{ if } m=n+1,\\
n+3 & \mbox{ if } m=n+2.
\end{array}\right.$$
\end{thm}

\begin{proof}
If $m=n+1$, the thesis follows from observing that the homogenization of a sequence of $n+1$ generic inhomogeneous polynomials in $n$ variables is a sequence of $n+1$ generic homogeneous polynomials in $n+1$ variables, hence it is a regular sequence.
If $m=n+2$, the thesis follows from applying Theorem~\ref{generic_aci} to $\F^h$.
\end{proof}

The reason for working over an infinite field in Theorem~\ref{generic_aci} and Theorem~\ref{thm:inhomog} is that the notion of generic polynomials is defined only over infinite fields. 
Over finite fields, one can no longer speak of generic polynomials. Nevertheless, one obtains the same bounds for cryptographic semi-regular sequences over finite fields. 
The next theorem follows from applying Theorem~\ref{generic_aci_ff} to $\F^h$.

\begin{thm}\label{thm:inhomog_ff}
Let $\mathbb{K}$ be a finite field and let $\F=\{f_1,\ldots,f_m\}\subseteq R$ be a cryptographic semi-regular sequence of inhomogeneous polynomials of degrees $d_i=\deg(f_i)$, with $m\in\{n+1,n+2\}$. Let $I=(f_1,\ldots,f_m)$. If $m=n+2$, assume without loss of generality that $d_{n+2}\leq d_1+\cdots+d_{n+1}-n-1$. Let
$$r(n,d_1,\ldots,d_m)=\left\{\begin{array}{ll}
d_1+\cdots+d_{n+1}-n & \mbox{ if } m=n+1,\\
\left\lfloor\frac{d_1+\cdots+d_{n+2}-n-2}{2}\right\rfloor+1 & \mbox{ if } m=n+2.
\end{array}\right.$$
If $d_1=\cdots=d_m=2$, then 
$$r(n,2,\ldots,2)=\left\{\begin{array}{ll}
n+2 & \mbox{ if } m=n+1,\\
\left\lfloor\frac{n}{2}\right\rfloor +2 & \mbox{ if } m=n+2,
\end{array}\right.$$
and if $d_1=\cdots=d_{n+1}=3$, then 
$$r(n,3,\ldots,3)=\left\{\begin{array}{ll}
2n+3 & \mbox{ if } m=n+1,\\
n+3 & \mbox{ if } m=n+2.
\end{array}\right.$$
Then $$\maxgb(I)\leq r(n,d_1,\ldots,d_m).$$
If in addition $J=(\F^h)$ is in generic coordinates, then $$\sd(\F)\leq r(n,d_1,\ldots,d_m).$$
\end{thm}

We now give an example of inhomogeneous cryptographic semi-regular sequences coming from index-calculus.

\begin{ex}
Using PYTHON, we performed the index-calculus algorithm on elliptic curves defined over finite fields $\mathbb{F}_{q^n}$, $q$ a large prime number and $n\in\{3,4,5\}$. Following the approach of Joux-Vitse~\cite{AJ}, we tried to decompose a random point on the curve as a sum of $n-1$ points of the factor basis and obtained overdetermined systems of $n$ equations in $n-1$ variables. Almost all the systems we produced are inhomogeneous cryptographic semi-regular systems. When we homogenize them, we obtain regular sequences of $n$ polynomials in $n$ variables. Therefore, their solving degree is bounded by $d_1 +\cdots+ d_n-(n-1)$ where $d_1,\cdots , d_n$ are the degrees of the $n$ polynomials. 

Notice that, since systems of this kind usually have no solutions, it is natural to expect that their homogenizations also have no solutions. Since the homogenizations are systems of $n$ polynomials in $n$ variables, they have no solutions if and only if they are regular sequences.
\end{ex}

We conclude the section with two more bounds on the solving degree of cryptographic semi-regular sequences of inhomogeneous polynomials. They follow from applying Theorem~\ref{d-reg} and Corollary~\ref{cor:largerm} to $\F^h$.

\begin{thm}
Let $\F=\{f_1,\ldots,f_m\}$ be a cryptographic semi-regular sequence of inhomogeneous polynomials of degree $2$ in $n$ variables. Let $I=(\F)$ and let
\begin{align*}
r(m,n)=\begin{cases}
\left\lceil(5+n-\sqrt{5+n})/2\right\rceil\quad &\text{if } m=n+3,\vspace{.2cm}\\\vspace{.2cm}
\left\lceil (7+n-\sqrt{19+3n})/2\right\rceil\quad &\text{if } m=n+4,\\\vspace{.2cm}
\left\lceil(9+n-\sqrt{23+3n+\sqrt{2}\sqrt{170+45n+3n^2}})/2\right\rceil\quad &\text{if } m=n+5,\\
\left\lceil(11+n-\sqrt{45+5n+\sqrt{2}\sqrt{368+85n+5n^2}})/2\right\rceil\quad &\text{if } m\geq n+6.
\end{cases}
\end{align*}
then $$\maxgb(I)\leq r(m,n).$$ Assume in addition that $\F^h$ generates an ideal which is in generic coordinates. Then $$\sd(\F)\leq r(m,n).$$
\end{thm}

\begin{thm}
Let $m\geq n+2$ and let $\F=\{f_1,\ldots,f_m\}$ be a cryptographic semi-regular sequence of inhomogeneous polynomials of degree $3$ in $n$ variables. Let $I=(\F)$, then 
$$\maxgb(I)\leq n+3.$$
Assume in addition that $\F^h$ generates an ideal which is in generic coordinates. Then $$\sd(\F)\leq n+3.$$
\end{thm}

\section{A consequence of the Eisenbud-Green-Harris Conjecture}\label{sec:eisen}

In this section, we present a conjecture formulated by Eisenbud, Green, and Harris in~\cite{EGH}. The conjecture, if true, has implications on the solving degree of overdetermined systems of quadratic equations, which we explore in this section. Before we state the conjecture, we give a definition and establish some notation.

\begin{defn}
Let $\ell\geq 0$, $d>0$. The {\bf Macaulay expansion} of $\ell$ with respect to $d$ is 
$$\ell={\ell_d\choose d}+\cdots+{\ell_2\choose 2}+{\ell_1\choose 1}$$
where $\ell_d>\ell_{d-1}>\cdots>\ell_1\geq 0$.
Set $0^{(d)}:=0$ and $$\ell^{(d)}:={\ell_d\choose d+1}+\cdots+{\ell_2\choose 3}+{\ell_1\choose 2}.$$
\end{defn}

Notice that ${a\choose b}=0$ if $a<b$. 
It is easy to show that the Macaulay expansion of $\ell$ with respect to $d$ exists and is unique.

\begin{ex}
The Macaulay expansion of $8$ with respect to $3$ is $$8={4\choose 3}+{3\choose 2}+{1\choose 1},\; 
\mbox{ hence }\; 8^{(3)}={4\choose 4}+{3\choose 3}+{1\choose 2}=2.$$
The Macaulay expansion of $10$ with respect to $3$ is $$10={5\choose 3}+{1\choose 2}+{0\choose 1},\;
\mbox{ hence }\; 10^{(3)}={5\choose 4}+{1\choose 3}+{0\choose 2}=5.$$ 
\end{ex}

The Eisenbud-Green-Harris Conjecture is a well-known conjecture in algebraic geometry, which has many equivalent formulations. The one that best suits our purpose is the following one.

\begin{conj}[{\cite[Conjecture $(V_m)$]{EGH}}]\label{conj:egh}
Assume that $f_1,\ldots,f_m\in R$ are homogeneous polynomials and that $I=(f_1,\ldots,f_m)$ contains a regular sequence of $n$ quadratic polynomials. Then $$H_{R/I}(d+1)\leq H_{R/I}(d)^{(d)}$$ for all $d>0$.
\end{conj}

In the next theorems, we explore the implications of Conjecture~\ref{conj:egh} on the solving degree of systems of quadratic equations.

\begin{thm}\label{thm:egh}
Assume that Conjecture~\ref{conj:egh} holds. Let $\F=\{f_1,\ldots,f_m\}\subseteq R$ be a system of homogeneous linearly independent quadratic polynomials such that $I=(f_1,\ldots,f_m)$ is Artinian. Let $\alpha$ be the unique integer such that $${n+1\choose 2}-{n-\alpha\choose 2}=\sum_{i=n-\alpha}^n i<m\leq \sum_{i=n-\alpha-1}^n i={n+1\choose 2}-{n-\alpha-1\choose 2}.$$ 
Then $$\maxgb(I)\leq n-\alpha.$$ 
If in addition $I=(f_1,\ldots,f_m)$ is in generic coordinates, then 
$$\sd(\F)\leq n-\alpha.$$ 
\end{thm}

\begin{proof}
Let $\mu:=\sum_{i=n-\alpha-1}^n i-m$. By assumption $0\leq\mu<n-\alpha-1$.
One has $$\dim_{\mathbb{K}}(R/I)_2={n+1\choose 2}-m={n-\alpha-1\choose 2}+\mu,$$
where the first equality follows from the linear independence of $f_1,\ldots,f_m$ and the second from the binomial formula for the sum of the first $n$ positive integers.  
Assuming that Conjecture~\ref{conj:egh} holds, we have $$\dim_{\mathbb{K}}(R/I)_3\leq\dim_{\mathbb{K}}(R/I)_2^{(2)}={n-\alpha-1\choose 3}+{\mu\choose 2}$$ hence $$\dim_{\mathbb{K}}(R/I)_4\leq\dim_{\mathbb{K}}(R/I)_3^{(3)}\leq {n-\alpha-1\choose 4}+{\mu\choose 3}.$$ After repeatedly applying the inequality of Conjecture~\ref{conj:egh}, we obtain
$$\dim_{\mathbb{K}}(R/I)_{n-\alpha}\leq {n-\alpha-1\choose n-\alpha}+{\mu\choose n-\alpha-1}=0,$$
since $\mu<n-\alpha-1$. Therefore $\reg(I)\leq n-\alpha$ and we conclude by~\cite[Theorem~3.22, Proposition~4.5, and Remark~4.6]{CG17}.
\end{proof}

Notice that Theorem~\ref{thm:egh} does not only apply to regular or semi-regular sequences, but more in general to {\bf any system of equations which contains a regular sequence of $n$ quadratic equations}.

\begin{rmk}
For $m=n$, Theorem~\ref{thm:egh} recovers the Macaulay bound. For $m>n$, Theorem~\ref{thm:egh} provides a better bound than the Macaulay bound.
\end{rmk}

In spite of its simplicity, Theorem~\ref{thm:egh} has nontrivial consequences. For example, it implies the following bound on the complexity of solving systems of quadratic equations over fields of characteristics 2. The assertion follows immediately from Theorem~\ref{thm:egh}.  Due to its length and technicality, we omit the description of Weil descent from this paper. Interested readers should refer to \cite[Chapter 7]{cohen2005}.

\begin{cor}\label{cor:egh}
Assume that Conjecture~\ref{conj:egh} holds. 
Let $\F'$ be a system of $m$ homogeneous quadratic equations in $\FF_{2^d}[x_1,\ldots,x_n]$ such that the ideal 
$(\F')$ is Artinian.
Choose a basis of $\FF_{2^d}$ over $\FF_2$ and let $\F$ be the system which contains the Weil descent of the equations of $\F'$ and the field equations of $\FF_2$. Let $I=(\F)$ and let $\ell=\dim_{\mathbb{K}}(I_2)$. 
Let $a$ be the unique integer such that $\sum_{i=nd-a}^{nd} i<\ell\leq \sum_{i=nd-a-1}^{nd} i.$
Then  $$\maxgb(I)\leq nd-a.$$
If in addition $I$ is in generic coordinates, then $$\sd(\F)\leq nd-a.$$
\end{cor}

Theorem~\ref{thm:egh} can be used to bound the solving degree of an inhomogeneous system of equations as follows. 

\begin{thm}\label{thm:egh_inhomog}
Assume that Conjecture~\ref{conj:egh} holds. Let $\F=\{f_1,\ldots,f_m\}\subseteq R$ be a system of inhomogeneous linearly independent quadratic polynomials such that $J=(\F^h)\subseteq S$ is Artinian. 
Let $\alpha$ be the unique integer such that $${n+2\choose 2}-{n+1-\alpha\choose 2}=\sum_{i=n+1-\alpha}^{n+1} i<m\leq \sum_{i=n-\alpha}^{n+1} i={n+2\choose 2}-{n-\alpha\choose 2}.$$ 
Then $$\maxgb(I)\leq n+1-\alpha.$$ 
If in addition $J$ is in generic coordinates, then 
$$\sd(\F)\leq n+1-\alpha.$$ 
\end{thm}

Notice that the assumption that $J$ is Artinian is verified, e.g., when $f_1,\ldots,f_{n+1}$ are generic polynomials.

\begin{cor}\label{cor:egh_inhomog}
Assume that Conjecture~\ref{conj:egh} holds. Let $\mathbb{K}$ be an infinite field and let $\F=\{f_1,\ldots,f_m\}\subseteq R$ be a system of inhomogeneous linearly independent quadratic polynomials such that $f_1,\ldots,f_{n+1}$ are generic. Let $\alpha$ be the unique integer such that $${n+2\choose 2}-{n+1-\alpha\choose 2}=\sum_{i=n+1-\alpha}^{n+1} i<m\leq \sum_{i=n-\alpha}^{n+1} i={n+2\choose 2}-{n-\alpha\choose 2}.$$ 
Then $$\maxgb(I)\leq n+1-\alpha.$$ 
If in addition $J=(\F^h)$ is in generic coordinates, then 
$$\sd(\F)\leq n+1-\alpha.$$ 
\end{cor}

Finally, one can obtain a result for inhomogeneous systems, along the lines of Corollary~\ref{cor:egh}.

\begin{cor}\label{cor:egh_inhomog2}
Assume that Conjecture~\ref{conj:egh} holds. Let $\F'$ be a system of $m$ quadratic equations in $\FF_{2^d}[x_1,\ldots,x_n]$ such that the ideal $(\F'^h)$ is Artinian. Choose a basis of $\FF_{2^d}$ over $\FF_2$ and let $\F$ be the system obtained from $\F'$ by Weil descent. If $\F$ contains $\ell$ linearly independent quadratic equations, let $a$ be the unique integer such that 
$$\sum_{i=(n+1)d-a}^{(n+1)d} i<\ell\leq \sum_{i=(n+1)d-a-1}^{(n+1)d} i.$$
Then  $$\maxgb(I)\leq (n+1)d-a.$$
If in addition $J=(\F^h)$ is in generic coordinates, then $$\sd(\F)\leq (n+1)d-a.$$
\end{cor}

\subsection{Limits to the applicability of Theorem~\ref{thm:egh} and relation with the degree of regularity}\label{sect:eghdreg}

This subsection is devoted to a discussion on the applicability of Theorem~\ref{thm:egh} and the other results of Section~\ref{sec:eisen} to systems arising in multivariate cryptography. We also establish a connection between the applicability of these results and the correctness of using the degree of regularity as a proxy for the solving degree. We use some results and standard arguments from commutative algebra, which we recall on a need-to-know basis, for the benefit of those who are not familiar with the commutative algebra literature. For the definition of degree of regularity and a discussion of its relation to the Castelnuovo-Mumford regularity, we refer the reader to~\cite[Section 4.1]{CG17}.

Throughout this subsection, we concentrate on the following situation. 

\begin{notat}\label{notat:lastsect}
Let $\mathbb{K}$ be a finite field and let $\F=\{f_1,\ldots,f_m\}\subseteq R$ be an inhomogeneous system of equations. Let $\F^h\subseteq S=R[t]$ be the system obtained by homogenizing the equations of $\F$ with respect to a new variable $t$. Let $I=(\F)\subseteq R$ be the ideal generated by $\F$ and let $J=(\F^h)\subseteq S$ be the ideal generated by $\F^h$. 
\end{notat}

\begin{ass}\label{ass:lastsect}
We assume that $\F$ contains the field equations and that it has at least one solution.
\end{ass}

Assumption~\ref{ass:lastsect} is satisfied for most systems associated to multivariate cryptosystems or multivariate signature schemes, after possibly adding the field equations to the system. It is not satisfied for most systems arising in the relation-finding phase of index-calculus algorithms to solve the elliptic and hyperelliptic curve discrete logarithm problem, and more in general the discrete logarithm problem on abelian variaties~\cite{G09}, as most of these systems have no solutions.

By~\cite[Theorem~3.23 and Theorem~3.26]{CG17}, in order to bound the solving degree of $\F$ it suffices to bound the regularity of $J$. Theorem~\ref{thm:egh} provides a bound on the regularity of certain Artinian ideals.
However, whenever $\F$ has a solution, $J$ is not Artinian. In fact, if $J$ is Artinian, then the only solution of $\F^h$ is the trivial solution $x_1=\cdots=x_n=t=0$. In particular, it has no solutions with $t=1$.
Since Theorem~\ref{thm:egh} applies to an Artinian ideal and $J$ is not Artinian, it would be natural to try to apply it to an Artinian reduction of $J$. In this subsection, we argue that {\bf it is unlikely that $J$ has an Artinian reduction}.

Before we go into the details of why this is the case, we need to recall a few definitions and facts from commutative algebra. Let $$J:(x_1,\ldots,x_n,t)^\infty=\{f\in S\mid \mbox{ there is a } d\geq 0 \mbox{ s.t. } fS_d\subseteq J\}$$ be the {\bf saturation} of $J$ with respect to the homogeneous maximal ideal of $S$. Here we denote by $fS_d=\{fg\mid g\in S_d\}$.
Let $$I^h=(f^h\mid f\in (f_1,\ldots,f_m))$$ be the {\bf homogenization} of the ideal $I$.
We stress that $I^h\supseteq J$, but they are not always equal. 

The following theorem gives equivalent conditions to $J$ having an Artinian reduction and provides a connection with the degree of regularity. Its relevance follows from observing that, whenever the equivalent conditions of the theorem hold, the solving degree of the system $\F$ can be estimated by applying Theorem~\ref{thm:egh} to $\F^{\ttop}$. 

\begin{thm}\label{thm:artinianred}
Adopt Notation~\ref{notat:lastsect} and assume that $\F$ satisfies Assumption~\ref{ass:lastsect}.
The following are equivalent:
\begin{enumerate}
\item $J$ has an Artinian reduction,
\item $(\F^{\ttop})$ is an Artinian reduction of $J$,
\item $t\nmid 0$ modulo $J$,
\item $J=I^h$.
\end{enumerate}
If the equivalent conditions hold, then $$\sd(\F)\leq\reg(\F^{\ttop})=d_{\reg}(\F).$$ 
\end{thm} 

\begin{proof}
Notice that the standard isomorphism between $S/(t)$ and $R$ sends $J+(t)/(t)$ to $(\F^{\ttop})$. Therefore, (2) and (3) are equivalent by Definition~\ref{defn:artinred}.
Clearly, (2) implies (1). 

We now show that (1) implies (3). 
Suppose that $J$ has an Artinian reduction. We claim that this implies that $J=J:(x_1,\ldots,x_n,t)^\infty$. By Definition~\ref{defn:artinred}, there is a homogeneous linear form $\ell\in S_1$ such that $\ell\nmid 0$ modulo $J$. Recall that, by definition of nonzero divisor, if $f\ell\in J$ for some $f\in S$, then $f\in J$. Let $f\in J:(x_1,\ldots,x_n,t)^\infty$ and let $d\geq 0$ be the least non-negative integer such that $fS_d\subseteq J$. If $d\geq 1$, then $f\ell S_{d-1}\subseteq fS_d\subseteq J$ and $\ell\nmid 0$ modulo $J$, hence $fS_{d-1}\subseteq J$. The latter contradicts the minimality of $d$, so it must be $d=0$, i.e., $f\in J$. Since $J\subseteq J:(x_1,\ldots,x_n,t)^\infty$ by the definition of saturation, we conclude that $J=J:(x_1,\ldots,x_n,t)^\infty$.
As shown in the proof of~\cite[Theorem~3.26]{CG17}, since $\F$ contains the field equations, then $t\nmid 0$ modulo $J:(x_1,\ldots,x_n,t)^\infty$. Since $J=J:(x_1,\ldots,x_n,t)^\infty$, then $t\nmid 0$ modulo $J$. This establishes the equivalence of conditions (1), (2), and (3).

Equivalence of (3) and (4) can be shown as follows. Consider the saturation of $J$ with respect to $t$, that is $$J:t^\infty=\{f\in S\mid \mbox{ there is a } d\geq 0 \mbox{ s.t. } ft^d\in J\}.$$ One has $$J\subseteq J:t^\infty=I^h,$$ where the inclusion follows from the definition and the equality by~\cite[Corollary~4.3.8]{KR2}. 
Therefore, (4) is equivalent to $J:t^\infty\subseteq J$. To check that the inclusion $J:t^\infty\subseteq J$ is equivalent to (3), let $f\in J:t^\infty$ and let $d$ be the least non-negative integer such that $ft^d\in J$. If $t\nmid 0$ modulo $J$ and $d\geq 1$, $ft^d\in J$ implies $ft^{d-1}\in J$, contradicting the minimality of $d$. Therefore it must be $d=0$, i.e., $f\in J$, showing that $J:t^\infty\subseteq J$. Conversely, if $J:t^\infty\subseteq J$, then $ft\in J$ implies $f\in J:t^\infty\subseteq J$, hence $t\nmid 0$ modulo $J$. 

Suppose now that the equivalent conditions hold. One has
$$\sd(\F)\leq \reg (J)=\reg(\F^{\ttop})=d_{\reg}(\F),$$
where the inequality follows from~\cite[Theorem~3.23 and Theorem~3.26]{CG17}, the first equality from Remark~\ref{rmk:artinred}, since $(\F^{\ttop})$ is an Artinian reduction of $J$, and the second equality from the definition of $d_{\reg}(\F)$.
\end{proof}

Finally, we give some evidence that the equivalent conditions of Theorem~\ref{thm:artinianred} are unlikely to hold, for a system $\mathcal{F}$ which satisfies Assumption~\ref{ass:lastsect}. One of the equivalent conditions of Theorem~\ref{thm:artinianred} is the equality $J=I^h$. As far as we know, the most general situation when $J=I^h$ is when {\bf $\F$ is a Gr\"obner basis of $I$ with respect to a degree-compatible term order}, as shown in~\cite[Proposition~4.3.21]{KR2}. It is clear that it is unlikely that the polynomials of $\F$ are a Gr\"obner basis, at least for systems $\F$ associated to multivariate cryptosystems or multivariate signature schemes.

\bigskip 

\addresseshere

\bigskip 

\appendix
\section{Values of $r(n+k,n)$ for $2 \le k,n \le 100$}

In this appendix, we provide the Castelnuovo-Mumford regularity $r(n+k,n)$ of the ideal generated by a cryptographic semi-regular system of $n+k$ homogeneous quadratic equations in $n$ variables. The value of $r(n+k,n)$ bounds the solving degree of cryptographic semi-regular systems of $n+k$ homogeneous polynomials in $n$ variables or $n+k$ inhomogeneous polynomials in $n-1$ variables (under the assumption that the system is in generic coordinates, as discussed in Section~\ref{semi-regular}). We have computed the value of $r(n+k,n)$ for $2 \le k,n \le 500$, however here we provide only the values for $2 \le k,n \le 100$. A table containing all the values that we computed can be found at~\url{bit.ly/wine-3}.

\begin{table}[ht]
	\hspace{-6cm}
	\caption{$r(n+k,n)$ for $2 \le k \le 100, 2\le n \le 26$} 
	\centering 
	\fontsize{6}{6}\selectfont
	\begin{tabular}{c|ccccccccccccccccccccccccc}
		\hline\hline
		k/n&2&3&4&5&6&7&8&9&10&11&12&13&14&15&16&17&18&19&20&21&22&23&24&25&26\\
		\hline
		2&2&3&3&3&4&4&5&5&6&6&6&7&7&8&8&9&9&10&10&10&11&11&12&12&13\\
		3&2&2&3&3&4&4&4&5&5&5&6&6&7&7&7&8&8&9&9&10&10&10&11&11&12\\
		4&2&2&3&3&3&4&4&4&5&5&5&6&6&7&7&7&8&8&8&9&9&10&10&10&11\\
		5&2&2&3&3&3&3&4&4&4&5&5&5&6&6&7&7&7&8&8&8&9&9&9&10&10\\
		6&2&2&2&3&3&3&4&4&4&5&5&5&6&6&6&7&7&7&8&8&8&9&9&9&10\\
		7&2&2&2&3&3&3&3&4&4&4&5&5&5&6&6&6&7&7&7&8&8&8&9&9&9\\
		8&2&2&2&3&3&3&3&4&4&4&4&5&5&5&6&6&6&7&7&7&8&8&8&9&9\\
		9&2&2&2&3&3&3&3&4&4&4&4&5&5&5&5&6&6&6&7&7&7&8&8&8&9\\
		10&2&2&2&2&3&3&3&3&4&4&4&4&5&5&5&6&6&6&6&7&7&7&8&8&8\\
		11&2&2&2&2&3&3&3&3&4&4&4&4&5&5&5&5&6&6&6&6&7&7&7&8&8\\
		12&2&2&2&2&3&3&3&3&3&4&4&4&4&5&5&5&5&6&6&6&7&7&7&7&8\\
		13&2&2&2&2&3&3&3&3&3&4&4&4&4&5&5&5&5&6&6&6&6&7&7&7&7\\
		14&2&2&2&2&3&3&3&3&3&4&4&4&4&4&5&5&5&5&6&6&6&6&7&7&7\\
		15&2&2&2&2&2&3&3&3&3&3&4&4&4&4&5&5&5&5&6&6&6&6&7&7&7\\
		16&2&2&2&2&2&3&3&3&3&3&4&4&4&4&5&5&5&5&5&6&6&6&6&7&7\\
		17&2&2&2&2&2&3&3&3&3&3&4&4&4&4&4&5&5&5&5&6&6&6&6&7&7\\
		18&2&2&2&2&2&3&3&3&3&3&4&4&4&4&4&5&5&5&5&5&6&6&6&6&7\\
		19&2&2&2&2&2&3&3&3&3&3&3&4&4&4&4&4&5&5&5&5&6&6&6&6&7\\
		20&2&2&2&2&2&3&3&3&3&3&3&4&4&4&4&4&5&5&5&5&5&6&6&6&6\\
		21&2&2&2&2&2&2&3&3&3&3&3&4&4&4&4&4&5&5&5&5&5&6&6&6&6\\
		22&2&2&2&2&2&2&3&3&3&3&3&3&4&4&4&4&4&5&5&5&5&5&6&6&6\\
		23&2&2&2&2&2&2&3&3&3&3&3&3&4&4&4&4&4&5&5&5&5&5&6&6&6\\
		24&2&2&2&2&2&2&3&3&3&3&3&3&4&4&4&4&4&5&5&5&5&5&6&6&6\\
		25&2&2&2&2&2&2&3&3&3&3&3&3&4&4&4&4&4&4&5&5&5&5&5&6&6\\
		26&2&2&2&2&2&2&3&3&3&3&3&3&3&4&4&4&4&4&5&5&5&5&5&6&6\\
		27&2&2&2&2&2&2&3&3&3&3&3&3&3&4&4&4&4&4&5&5&5&5&5&5&6\\
		28&2&2&2&2&2&2&2&3&3&3&3&3&3&4&4&4&4&4&4&5&5&5&5&5&6\\
		29&2&2&2&2&2&2&2&3&3&3&3&3&3&4&4&4&4&4&4&5&5&5&5&5&6\\
		30&2&2&2&2&2&2&2&3&3&3&3&3&3&4&4&4&4&4&4&5&5&5&5&5&5\\
		31&2&2&2&2&2&2&2&3&3&3&3&3&3&3&4&4&4&4&4&4&5&5&5&5&5\\
		32&2&2&2&2&2&2&2&3&3&3&3&3&3&3&4&4&4&4&4&4&5&5&5&5&5\\
		33&2&2&2&2&2&2&2&3&3&3&3&3&3&3&4&4&4&4&4&4&5&5&5&5&5\\
		34&2&2&2&2&2&2&2&3&3&3&3&3&3&3&4&4&4&4&4&4&5&5&5&5&5\\
		35&2&2&2&2&2&2&2&3&3&3&3&3&3&3&3&4&4&4&4&4&4&5&5&5&5\\
		36&2&2&2&2&2&2&2&2&3&3&3&3&3&3&3&4&4&4&4&4&4&5&5&5&5\\
		37&2&2&2&2&2&2&2&2&3&3&3&3&3&3&3&4&4&4&4&4&4&5&5&5&5\\
		38&2&2&2&2&2&2&2&2&3&3&3&3&3&3&3&4&4&4&4&4&4&4&5&5&5\\
		39&2&2&2&2&2&2&2&2&3&3&3&3&3&3&3&4&4&4&4&4&4&4&5&5&5\\
		40&2&2&2&2&2&2&2&2&3&3&3&3&3&3&3&3&4&4&4&4&4&4&5&5&5\\
		41&2&2&2&2&2&2&2&2&3&3&3&3&3&3&3&3&4&4&4&4&4&4&5&5&5\\
		42&2&2&2&2&2&2&2&2&3&3&3&3&3&3&3&3&4&4&4&4&4&4&4&5&5\\
		43&2&2&2&2&2&2&2&2&3&3&3&3&3&3&3&3&4&4&4&4&4&4&4&5&5\\
		44&2&2&2&2&2&2&2&2&3&3&3&3&3&3&3&3&4&4&4&4&4&4&4&5&5\\
		45&2&2&2&2&2&2&2&2&2&3&3&3&3&3&3&3&4&4&4&4&4&4&4&5&5\\
		46&2&2&2&2&2&2&2&2&2&3&3&3&3&3&3&3&3&4&4&4&4&4&4&4&5\\
		47&2&2&2&2&2&2&2&2&2&3&3&3&3&3&3&3&3&4&4&4&4&4&4&4&5\\
		48&2&2&2&2&2&2&2&2&2&3&3&3&3&3&3&3&3&4&4&4&4&4&4&4&5\\
		49&2&2&2&2&2&2&2&2&2&3&3&3&3&3&3&3&3&4&4&4&4&4&4&4&5\\
		50&2&2&2&2&2&2&2&2&2&3&3&3&3&3&3&3&3&4&4&4&4&4&4&4&4\\
		51&2&2&2&2&2&2&2&2&2&3&3&3&3&3&3&3&3&3&4&4&4&4&4&4&4\\
		52&2&2&2&2&2&2&2&2&2&3&3&3&3&3&3&3&3&3&4&4&4&4&4&4&4\\
		53&2&2&2&2&2&2&2&2&2&3&3&3&3&3&3&3&3&3&4&4&4&4&4&4&4\\
		54&2&2&2&2&2&2&2&2&2&3&3&3&3&3&3&3&3&3&4&4&4&4&4&4&4\\
		55&2&2&2&2&2&2&2&2&2&2&3&3&3&3&3&3&3&3&4&4&4&4&4&4&4\\
		56&2&2&2&2&2&2&2&2&2&2&3&3&3&3&3&3&3&3&4&4&4&4&4&4&4\\
		57&2&2&2&2&2&2&2&2&2&2&3&3&3&3&3&3&3&3&3&4&4&4&4&4&4\\
		58&2&2&2&2&2&2&2&2&2&2&3&3&3&3&3&3&3&3&3&4&4&4&4&4&4\\
		59&2&2&2&2&2&2&2&2&2&2&3&3&3&3&3&3&3&3&3&4&4&4&4&4&4\\
		60&2&2&2&2&2&2&2&2&2&2&3&3&3&3&3&3&3&3&3&4&4&4&4&4&4\\
		61&2&2&2&2&2&2&2&2&2&2&3&3&3&3&3&3&3&3&3&4&4&4&4&4&4\\
		62&2&2&2&2&2&2&2&2&2&2&3&3&3&3&3&3&3&3&3&4&4&4&4&4&4\\
		63&2&2&2&2&2&2&2&2&2&2&3&3&3&3&3&3&3&3&3&4&4&4&4&4&4\\
		64&2&2&2&2&2&2&2&2&2&2&3&3&3&3&3&3&3&3&3&3&4&4&4&4&4\\
		65&2&2&2&2&2&2&2&2&2&2&3&3&3&3&3&3&3&3&3&3&4&4&4&4&4\\
		66&2&2&2&2&2&2&2&2&2&2&2&3&3&3&3&3&3&3&3&3&4&4&4&4&4\\
		67&2&2&2&2&2&2&2&2&2&2&2&3&3&3&3&3&3&3&3&3&4&4&4&4&4\\
		68&2&2&2&2&2&2&2&2&2&2&2&3&3&3&3&3&3&3&3&3&4&4&4&4&4\\
		69&2&2&2&2&2&2&2&2&2&2&2&3&3&3&3&3&3&3&3&3&4&4&4&4&4\\
		70&2&2&2&2&2&2&2&2&2&2&2&3&3&3&3&3&3&3&3&3&3&4&4&4&4\\
		71&2&2&2&2&2&2&2&2&2&2&2&3&3&3&3&3&3&3&3&3&3&4&4&4&4\\
		72&2&2&2&2&2&2&2&2&2&2&2&3&3&3&3&3&3&3&3&3&3&4&4&4&4\\
		73&2&2&2&2&2&2&2&2&2&2&2&3&3&3&3&3&3&3&3&3&3&4&4&4&4\\
		74&2&2&2&2&2&2&2&2&2&2&2&3&3&3&3&3&3&3&3&3&3&4&4&4&4\\
		75&2&2&2&2&2&2&2&2&2&2&2&3&3&3&3&3&3&3&3&3&3&4&4&4&4\\
		76&2&2&2&2&2&2&2&2&2&2&2&3&3&3&3&3&3&3&3&3&3&4&4&4&4\\
		77&2&2&2&2&2&2&2&2&2&2&2&3&3&3&3&3&3&3&3&3&3&3&4&4&4\\
		78&2&2&2&2&2&2&2&2&2&2&2&2&3&3&3&3&3&3&3&3&3&3&4&4&4\\
		79&2&2&2&2&2&2&2&2&2&2&2&2&3&3&3&3&3&3&3&3&3&3&4&4&4\\
		80&2&2&2&2&2&2&2&2&2&2&2&2&3&3&3&3&3&3&3&3&3&3&4&4&4\\
		81&2&2&2&2&2&2&2&2&2&2&2&2&3&3&3&3&3&3&3&3&3&3&4&4&4\\
		82&2&2&2&2&2&2&2&2&2&2&2&2&3&3&3&3&3&3&3&3&3&3&4&4&4\\
		83&2&2&2&2&2&2&2&2&2&2&2&2&3&3&3&3&3&3&3&3&3&3&4&4&4\\
		84&2&2&2&2&2&2&2&2&2&2&2&2&3&3&3&3&3&3&3&3&3&3&4&4&4\\
		85&2&2&2&2&2&2&2&2&2&2&2&2&3&3&3&3&3&3&3&3&3&3&3&4&4\\
		86&2&2&2&2&2&2&2&2&2&2&2&2&3&3&3&3&3&3&3&3&3&3&3&4&4\\
		87&2&2&2&2&2&2&2&2&2&2&2&2&3&3&3&3&3&3&3&3&3&3&3&4&4\\
		88&2&2&2&2&2&2&2&2&2&2&2&2&3&3&3&3&3&3&3&3&3&3&3&4&4\\
		89&2&2&2&2&2&2&2&2&2&2&2&2&3&3&3&3&3&3&3&3&3&3&3&4&4\\
		90&2&2&2&2&2&2&2&2&2&2&2&2&3&3&3&3&3&3&3&3&3&3&3&4&4\\
		91&2&2&2&2&2&2&2&2&2&2&2&2&2&3&3&3&3&3&3&3&3&3&3&4&4\\
		92&2&2&2&2&2&2&2&2&2&2&2&2&2&3&3&3&3&3&3&3&3&3&3&3&4\\
		93&2&2&2&2&2&2&2&2&2&2&2&2&2&3&3&3&3&3&3&3&3&3&3&3&4\\
		94&2&2&2&2&2&2&2&2&2&2&2&2&2&3&3&3&3&3&3&3&3&3&3&3&4\\
		95&2&2&2&2&2&2&2&2&2&2&2&2&2&3&3&3&3&3&3&3&3&3&3&3&4\\
		96&2&2&2&2&2&2&2&2&2&2&2&2&2&3&3&3&3&3&3&3&3&3&3&3&4\\
		97&2&2&2&2&2&2&2&2&2&2&2&2&2&3&3&3&3&3&3&3&3&3&3&3&4\\
		98&2&2&2&2&2&2&2&2&2&2&2&2&2&3&3&3&3&3&3&3&3&3&3&3&4\\
		99&2&2&2&2&2&2&2&2&2&2&2&2&2&3&3&3&3&3&3&3&3&3&3&3&4\\
		100&2&2&2&2&2&2&2&2&2&2&2&2&2&3&3&3&3&3&3&3&3&3&3&3&3\\
		\hline
	\end{tabular}
\end{table}
\begin{table}[ht]
	\hspace{-6cm}
	\caption{$r(n+k,n)$ for $2 \le k \le 100, 27\le n \le 51$} 
	\centering 
	\fontsize{6}{6}\selectfont
	\begin{tabular}{c|ccccccccccccccccccccccccc}
		\hline\hline
		k/n&27&28&29&30&31&32&33&34&35&36&37&38&39&40&41&42&43&44&45&46&47&48&49&50&51\\
		\hline
		2&25&25&26&26&27&27&28&28&28&29&29&30&30&31&31&32&32&33&33&34&34&35&35&36&36\\
		3&23&23&24&24&25&25&26&26&26&27&27&28&28&29&29&30&30&31&31&31&32&32&33&33&34\\
		4&22&22&22&23&23&24&24&25&25&25&26&26&27&27&28&28&28&29&29&30&30&31&31&31&32\\
		5&20&21&21&22&22&22&23&23&24&24&25&25&25&26&26&27&27&27&28&28&29&29&30&30&30\\
		6&20&20&20&21&21&21&22&22&23&23&23&24&24&25&25&26&26&26&27&27&28&28&28&29&29\\
		7&19&19&19&20&20&21&21&21&22&22&23&23&23&24&24&25&25&25&26&26&27&27&27&28&28\\
		8&18&18&19&19&19&20&20&21&21&21&22&22&23&23&23&24&24&24&25&25&26&26&26&27&27\\
		9&17&18&18&18&19&19&20&20&20&21&21&21&22&22&23&23&23&24&24&24&25&25&26&26&26\\
		10&17&17&18&18&18&19&19&19&20&20&20&21&21&21&22&22&23&23&23&24&24&24&25&25&26\\
		11&16&17&17&17&18&18&18&19&19&19&20&20&20&21&21&22&22&22&23&23&23&24&24&24&25\\
		12&16&16&17&17&17&18&18&18&19&19&19&20&20&20&21&21&21&22&22&22&23&23&23&24&24\\
		13&15&16&16&16&17&17&17&18&18&18&19&19&19&20&20&20&21&21&21&22&22&22&23&23&24\\
		14&15&15&16&16&16&17&17&17&18&18&18&19&19&19&20&20&20&21&21&21&22&22&22&23&23\\
		15&15&15&15&16&16&16&17&17&17&18&18&18&18&19&19&19&20&20&20&21&21&21&22&22&22\\
		16&14&15&15&15&16&16&16&16&17&17&17&18&18&18&19&19&19&20&20&20&21&21&21&22&22\\
		17&14&14&15&15&15&16&16&16&16&17&17&17&18&18&18&19&19&19&20&20&20&21&21&21&22\\
		18&14&14&14&15&15&15&15&16&16&16&17&17&17&18&18&18&19&19&19&19&20&20&20&21&21\\
		19&13&14&14&14&15&15&15&15&16&16&16&17&17&17&18&18&18&18&19&19&19&20&20&20&21\\
		20&13&13&14&14&14&15&15&15&15&16&16&16&17&17&17&18&18&18&18&19&19&19&20&20&20\\
		21&13&13&13&14&14&14&15&15&15&15&16&16&16&17&17&17&18&18&18&18&19&19&19&20&20\\
		22&13&13&13&13&14&14&14&15&15&15&15&16&16&16&17&17&17&18&18&18&18&19&19&19&20\\
		23&12&13&13&13&14&14&14&14&15&15&15&15&16&16&16&17&17&17&17&18&18&18&19&19&19\\
		24&12&13&13&13&13&14&14&14&14&15&15&15&16&16&16&16&17&17&17&17&18&18&18&19&19\\
		25&12&12&13&13&13&13&14&14&14&14&15&15&15&16&16&16&16&17&17&17&17&18&18&18&19\\
		26&12&12&12&13&13&13&13&14&14&14&14&15&15&15&16&16&16&16&17&17&17&17&18&18&18\\
		27&12&12&12&12&13&13&13&13&14&14&14&15&15&15&15&16&16&16&16&17&17&17&18&18&18\\
		28&11&12&12&12&12&13&13&13&14&14&14&14&15&15&15&15&16&16&16&16&17&17&17&18&18\\
		29&11&12&12&12&12&13&13&13&13&14&14&14&14&15&15&15&15&16&16&16&16&17&17&17&18\\
		30&11&11&12&12&12&12&13&13&13&13&14&14&14&14&15&15&15&15&16&16&16&16&17&17&17\\
		31&11&11&11&12&12&12&12&13&13&13&13&14&14&14&14&15&15&15&15&16&16&16&17&17&17\\
		32&11&11&11&12&12&12&12&13&13&13&13&14&14&14&14&15&15&15&15&16&16&16&16&17&17\\
		33&11&11&11&11&12&12&12&12&13&13&13&13&14&14&14&14&15&15&15&15&16&16&16&16&17\\
		34&11&11&11&11&11&12&12&12&12&13&13&13&13&14&14&14&14&15&15&15&15&16&16&16&16\\
		35&10&11&11&11&11&12&12&12&12&13&13&13&13&13&14&14&14&14&15&15&15&15&16&16&16\\
		36&10&11&11&11&11&11&12&12&12&12&13&13&13&13&14&14&14&14&15&15&15&15&16&16&16\\
		37&10&10&11&11&11&11&12&12&12&12&12&13&13&13&13&14&14&14&14&15&15&15&15&16&16\\
		38&10&10&10&11&11&11&11&12&12&12&12&13&13&13&13&13&14&14&14&14&15&15&15&15&16\\
		39&10&10&10&11&11&11&11&11&12&12&12&12&13&13&13&13&14&14&14&14&14&15&15&15&15\\
		40&10&10&10&10&11&11&11&11&12&12&12&12&12&13&13&13&13&14&14&14&14&15&15&15&15\\
		41&10&10&10&10&11&11&11&11&11&12&12&12&12&13&13&13&13&13&14&14&14&14&15&15&15\\
		42&10&10&10&10&10&11&11&11&11&12&12&12&12&12&13&13&13&13&14&14&14&14&14&15&15\\
		43&10&10&10&10&10&11&11&11&11&11&12&12&12&12&13&13&13&13&13&14&14&14&14&15&15\\
		44&9&10&10&10&10&10&11&11&11&11&12&12&12&12&12&13&13&13&13&13&14&14&14&14&15\\
		45&9&10&10&10&10&10&11&11&11&11&11&12&12&12&12&12&13&13&13&13&14&14&14&14&14\\
		46&9&9&10&10&10&10&10&11&11&11&11&11&12&12&12&12&13&13&13&13&13&14&14&14&14\\
		47&9&9&10&10&10&10&10&11&11&11&11&11&12&12&12&12&12&13&13&13&13&14&14&14&14\\
		48&9&9&9&10&10&10&10&10&11&11&11&11&11&12&12&12&12&13&13&13&13&13&14&14&14\\
		49&9&9&9&10&10&10&10&10&11&11&11&11&11&12&12&12&12&12&13&13&13&13&13&14&14\\
		50&9&9&9&9&10&10&10&10&10&11&11&11&11&11&12&12&12&12&13&13&13&13&13&14&14\\
		51&9&9&9&9&10&10&10&10&10&11&11&11&11&11&12&12&12&12&12&13&13&13&13&13&14\\
		52&9&9&9&9&9&10&10&10&10&10&11&11&11&11&11&12&12&12&12&12&13&13&13&13&14\\
		53&9&9&9&9&9&10&10&10&10&10&11&11&11&11&11&12&12&12&12&12&13&13&13&13&13\\
		54&9&9&9&9&9&10&10&10&10&10&10&11&11&11&11&11&12&12&12&12&12&13&13&13&13\\
		55&9&9&9&9&9&9&10&10&10&10&10&11&11&11&11&11&12&12&12&12&12&13&13&13&13\\
		56&8&9&9&9&9&9&10&10&10&10&10&10&11&11&11&11&11&12&12&12&12&12&13&13&13\\
		57&8&9&9&9&9&9&9&10&10&10&10&10&11&11&11&11&11&12&12&12&12&12&13&13&13\\
		58&8&8&9&9&9&9&9&10&10&10&10&10&10&11&11&11&11&11&12&12&12&12&12&13&13\\
		59&8&8&9&9&9&9&9&9&10&10&10&10&10&11&11&11&11&11&12&12&12&12&12&13&13\\
		60&8&8&9&9&9&9&9&9&10&10&10&10&10&11&11&11&11&11&11&12&12&12&12&12&13\\
		61&8&8&8&9&9&9&9&9&10&10&10&10&10&10&11&11&11&11&11&12&12&12&12&12&13\\
		62&8&8&8&9&9&9&9&9&9&10&10&10&10&10&11&11&11&11&11&11&12&12&12&12&12\\
		63&8&8&8&8&9&9&9&9&9&10&10&10&10&10&10&11&11&11&11&11&12&12&12&12&12\\
		64&8&8&8&8&9&9&9&9&9&9&10&10&10&10&10&11&11&11&11&11&11&12&12&12&12\\
		65&8&8&8&8&9&9&9&9&9&9&10&10&10&10&10&10&11&11&11&11&11&12&12&12&12\\
		66&8&8&8&8&8&9&9&9&9&9&10&10&10&10&10&10&11&11&11&11&11&12&12&12&12\\
		67&8&8&8&8&8&9&9&9&9&9&9&10&10&10&10&10&11&11&11&11&11&11&12&12&12\\
		68&8&8&8&8&8&9&9&9&9&9&9&10&10&10&10&10&10&11&11&11&11&11&12&12&12\\
		69&8&8&8&8&8&8&9&9&9&9&9&9&10&10&10&10&10&11&11&11&11&11&11&12&12\\
		70&8&8&8&8&8&8&9&9&9&9&9&9&10&10&10&10&10&10&11&11&11&11&11&12&12\\
		71&8&8&8&8&8&8&9&9&9&9&9&9&10&10&10&10&10&10&11&11&11&11&11&11&12\\
		72&8&8&8&8&8&8&8&9&9&9&9&9&9&10&10&10&10&10&10&11&11&11&11&11&12\\
		73&7&8&8&8&8&8&8&9&9&9&9&9&9&10&10&10&10&10&10&11&11&11&11&11&11\\
		74&7&8&8&8&8&8&8&9&9&9&9&9&9&10&10&10&10&10&10&11&11&11&11&11&11\\
		75&7&8&8&8&8&8&8&8&9&9&9&9&9&9&10&10&10&10&10&10&11&11&11&11&11\\
		76&7&7&8&8&8&8&8&8&9&9&9&9&9&9&10&10&10&10&10&10&11&11&11&11&11\\
		77&7&7&8&8&8&8&8&8&9&9&9&9&9&9&9&10&10&10&10&10&10&11&11&11&11\\
		78&7&7&8&8&8&8&8&8&8&9&9&9&9&9&9&10&10&10&10&10&10&11&11&11&11\\
		79&7&7&7&8&8&8&8&8&8&9&9&9&9&9&9&10&10&10&10&10&10&11&11&11&11\\
		80&7&7&7&8&8&8&8&8&8&9&9&9&9&9&9&9&10&10&10&10&10&10&11&11&11\\
		81&7&7&7&8&8&8&8&8&8&8&9&9&9&9&9&9&10&10&10&10&10&10&11&11&11\\
		82&7&7&7&7&8&8&8&8&8&8&9&9&9&9&9&9&10&10&10&10&10&10&11&11&11\\
		83&7&7&7&7&8&8&8&8&8&8&9&9&9&9&9&9&9&10&10&10&10&10&10&11&11\\
		84&7&7&7&7&8&8&8&8&8&8&8&9&9&9&9&9&9&10&10&10&10&10&10&11&11\\
		85&7&7&7&7&8&8&8&8&8&8&8&9&9&9&9&9&9&10&10&10&10&10&10&10&11\\
		86&7&7&7&7&7&8&8&8&8&8&8&9&9&9&9&9&9&9&10&10&10&10&10&10&11\\
		87&7&7&7&7&7&8&8&8&8&8&8&8&9&9&9&9&9&9&10&10&10&10&10&10&11\\
		88&7&7&7&7&7&8&8&8&8&8&8&8&9&9&9&9&9&9&10&10&10&10&10&10&10\\
		89&7&7&7&7&7&7&8&8&8&8&8&8&9&9&9&9&9&9&9&10&10&10&10&10&10\\
		90&7&7&7&7&7&7&8&8&8&8&8&8&8&9&9&9&9&9&9&10&10&10&10&10&10\\
		91&7&7&7&7&7&7&8&8&8&8&8&8&8&9&9&9&9&9&9&10&10&10&10&10&10\\
		92&7&7&7&7&7&7&8&8&8&8&8&8&8&9&9&9&9&9&9&9&10&10&10&10&10\\
		93&7&7&7&7&7&7&7&8&8&8&8&8&8&8&9&9&9&9&9&9&10&10&10&10&10\\
		94&7&7&7&7&7&7&7&8&8&8&8&8&8&8&9&9&9&9&9&9&10&10&10&10&10\\
		95&7&7&7&7&7&7&7&8&8&8&8&8&8&8&9&9&9&9&9&9&9&10&10&10&10\\
		96&7&7&7&7&7&7&7&8&8&8&8&8&8&8&9&9&9&9&9&9&9&10&10&10&10\\
		97&7&7&7&7&7&7&7&7&8&8&8&8&8&8&8&9&9&9&9&9&9&10&10&10&10\\
		98&6&7&7&7&7&7&7&7&8&8&8&8&8&8&8&9&9&9&9&9&9&9&10&10&10\\
		99&6&7&7&7&7&7&7&7&8&8&8&8&8&8&8&9&9&9&9&9&9&9&10&10&10\\
		100&6&7&7&7&7&7&7&7&7&8&8&8&8&8&8&8&9&9&9&9&9&9&10&10&10\\
		\hline
	\end{tabular}
\end{table}

\begin{table}[ht]
	\hspace{-6cm}
	\caption{$r(n+k,n)$ for $2 \le k \le 100, 52\le n \le 76$} 
	\centering 
	\fontsize{6}{6}\selectfont
	\begin{tabular}{c|ccccccccccccccccccccccccc}
		\hline\hline
		k/n&52&53&54&55&56&57&58&59&60&61&62&63&64&65&66&67&68&69&70&71&72&73&74&75&76\\
		\hline
		2&25&25&26&26&27&27&28&28&28&29&29&30&30&31&31&32&32&33&33&34&34&35&35&36&36\\
		3&23&23&24&24&25&25&26&26&26&27&27&28&28&29&29&30&30&31&31&31&32&32&33&33&34\\
		4&22&22&22&23&23&24&24&25&25&25&26&26&27&27&28&28&28&29&29&30&30&31&31&31&32\\
		5&20&21&21&22&22&22&23&23&24&24&25&25&25&26&26&27&27&27&28&28&29&29&30&30&30\\
		6&20&20&20&21&21&21&22&22&23&23&23&24&24&25&25&26&26&26&27&27&28&28&28&29&29\\
		7&19&19&19&20&20&21&21&21&22&22&23&23&23&24&24&25&25&25&26&26&27&27&27&28&28\\
		8&18&18&19&19&19&20&20&21&21&21&22&22&23&23&23&24&24&24&25&25&26&26&26&27&27\\
		9&17&18&18&18&19&19&20&20&20&21&21&21&22&22&23&23&23&24&24&24&25&25&26&26&26\\
		10&17&17&18&18&18&19&19&19&20&20&20&21&21&21&22&22&23&23&23&24&24&24&25&25&26\\
		11&16&17&17&17&18&18&18&19&19&19&20&20&20&21&21&22&22&22&23&23&23&24&24&24&25\\
		12&16&16&17&17&17&18&18&18&19&19&19&20&20&20&21&21&21&22&22&22&23&23&23&24&24\\
		13&15&16&16&16&17&17&17&18&18&18&19&19&19&20&20&20&21&21&21&22&22&22&23&23&24\\
		14&15&15&16&16&16&17&17&17&18&18&18&19&19&19&20&20&20&21&21&21&22&22&22&23&23\\
		15&15&15&15&16&16&16&17&17&17&18&18&18&18&19&19&19&20&20&20&21&21&21&22&22&22\\
		16&14&15&15&15&16&16&16&16&17&17&17&18&18&18&19&19&19&20&20&20&21&21&21&22&22\\
		17&14&14&15&15&15&16&16&16&16&17&17&17&18&18&18&19&19&19&20&20&20&21&21&21&22\\
		18&14&14&14&15&15&15&15&16&16&16&17&17&17&18&18&18&19&19&19&19&20&20&20&21&21\\
		19&13&14&14&14&15&15&15&15&16&16&16&17&17&17&18&18&18&18&19&19&19&20&20&20&21\\
		20&13&13&14&14&14&15&15&15&15&16&16&16&17&17&17&18&18&18&18&19&19&19&20&20&20\\
		21&13&13&13&14&14&14&15&15&15&15&16&16&16&17&17&17&18&18&18&18&19&19&19&20&20\\
		22&13&13&13&13&14&14&14&15&15&15&15&16&16&16&17&17&17&18&18&18&18&19&19&19&20\\
		23&12&13&13&13&14&14&14&14&15&15&15&15&16&16&16&17&17&17&17&18&18&18&19&19&19\\
		24&12&13&13&13&13&14&14&14&14&15&15&15&16&16&16&16&17&17&17&17&18&18&18&19&19\\
		25&12&12&13&13&13&13&14&14&14&14&15&15&15&16&16&16&16&17&17&17&17&18&18&18&19\\
		26&12&12&12&13&13&13&13&14&14&14&14&15&15&15&16&16&16&16&17&17&17&17&18&18&18\\
		27&12&12&12&12&13&13&13&13&14&14&14&15&15&15&15&16&16&16&16&17&17&17&18&18&18\\
		28&11&12&12&12&12&13&13&13&14&14&14&14&15&15&15&15&16&16&16&16&17&17&17&18&18\\
		29&11&12&12&12&12&13&13&13&13&14&14&14&14&15&15&15&15&16&16&16&16&17&17&17&18\\
		30&11&11&12&12&12&12&13&13&13&13&14&14&14&14&15&15&15&15&16&16&16&16&17&17&17\\
		31&11&11&11&12&12&12&12&13&13&13&13&14&14&14&14&15&15&15&15&16&16&16&17&17&17\\
		32&11&11&11&12&12&12&12&13&13&13&13&14&14&14&14&15&15&15&15&16&16&16&16&17&17\\
		33&11&11&11&11&12&12&12&12&13&13&13&13&14&14&14&14&15&15&15&15&16&16&16&16&17\\
		34&11&11&11&11&11&12&12&12&12&13&13&13&13&14&14&14&14&15&15&15&15&16&16&16&16\\
		35&10&11&11&11&11&12&12&12&12&13&13&13&13&13&14&14&14&14&15&15&15&15&16&16&16\\
		36&10&11&11&11&11&11&12&12&12&12&13&13&13&13&14&14&14&14&15&15&15&15&16&16&16\\
		37&10&10&11&11&11&11&12&12&12&12&12&13&13&13&13&14&14&14&14&15&15&15&15&16&16\\
		38&10&10&10&11&11&11&11&12&12&12&12&13&13&13&13&13&14&14&14&14&15&15&15&15&16\\
		39&10&10&10&11&11&11&11&11&12&12&12&12&13&13&13&13&14&14&14&14&14&15&15&15&15\\
		40&10&10&10&10&11&11&11&11&12&12&12&12&12&13&13&13&13&14&14&14&14&15&15&15&15\\
		41&10&10&10&10&11&11&11&11&11&12&12&12&12&13&13&13&13&13&14&14&14&14&15&15&15\\
		42&10&10&10&10&10&11&11&11&11&12&12&12&12&12&13&13&13&13&14&14&14&14&14&15&15\\
		43&10&10&10&10&10&11&11&11&11&11&12&12&12&12&13&13&13&13&13&14&14&14&14&15&15\\
		44&9&10&10&10&10&10&11&11&11&11&12&12&12&12&12&13&13&13&13&13&14&14&14&14&15\\
		45&9&10&10&10&10&10&11&11&11&11&11&12&12&12&12&12&13&13&13&13&14&14&14&14&14\\
		46&9&9&10&10&10&10&10&11&11&11&11&11&12&12&12&12&13&13&13&13&13&14&14&14&14\\
		47&9&9&10&10&10&10&10&11&11&11&11&11&12&12&12&12&12&13&13&13&13&14&14&14&14\\
		48&9&9&9&10&10&10&10&10&11&11&11&11&11&12&12&12&12&13&13&13&13&13&14&14&14\\
		49&9&9&9&10&10&10&10&10&11&11&11&11&11&12&12&12&12&12&13&13&13&13&13&14&14\\
		50&9&9&9&9&10&10&10&10&10&11&11&11&11&11&12&12&12&12&13&13&13&13&13&14&14\\
		51&9&9&9&9&10&10&10&10&10&11&11&11&11&11&12&12&12&12&12&13&13&13&13&13&14\\
		52&9&9&9&9&9&10&10&10&10&10&11&11&11&11&11&12&12&12&12&12&13&13&13&13&14\\
		53&9&9&9&9&9&10&10&10&10&10&11&11&11&11&11&12&12&12&12&12&13&13&13&13&13\\
		54&9&9&9&9&9&10&10&10&10&10&10&11&11&11&11&11&12&12&12&12&12&13&13&13&13\\
		55&9&9&9&9&9&9&10&10&10&10&10&11&11&11&11&11&12&12&12&12&12&13&13&13&13\\
		56&8&9&9&9&9&9&10&10&10&10&10&10&11&11&11&11&11&12&12&12&12&12&13&13&13\\
		57&8&9&9&9&9&9&9&10&10&10&10&10&11&11&11&11&11&12&12&12&12&12&13&13&13\\
		58&8&8&9&9&9&9&9&10&10&10&10&10&10&11&11&11&11&11&12&12&12&12&12&13&13\\
		59&8&8&9&9&9&9&9&9&10&10&10&10&10&11&11&11&11&11&12&12&12&12&12&13&13\\
		60&8&8&9&9&9&9&9&9&10&10&10&10&10&11&11&11&11&11&11&12&12&12&12&12&13\\
		61&8&8&8&9&9&9&9&9&10&10&10&10&10&10&11&11&11&11&11&12&12&12&12&12&13\\
		62&8&8&8&9&9&9&9&9&9&10&10&10&10&10&11&11&11&11&11&11&12&12&12&12&12\\
		63&8&8&8&8&9&9&9&9&9&10&10&10&10&10&10&11&11&11&11&11&12&12&12&12&12\\
		64&8&8&8&8&9&9&9&9&9&9&10&10&10&10&10&11&11&11&11&11&11&12&12&12&12\\
		65&8&8&8&8&9&9&9&9&9&9&10&10&10&10&10&10&11&11&11&11&11&12&12&12&12\\
		66&8&8&8&8&8&9&9&9&9&9&10&10&10&10&10&10&11&11&11&11&11&12&12&12&12\\
		67&8&8&8&8&8&9&9&9&9&9&9&10&10&10&10&10&11&11&11&11&11&11&12&12&12\\
		68&8&8&8&8&8&9&9&9&9&9&9&10&10&10&10&10&10&11&11&11&11&11&12&12&12\\
		69&8&8&8&8&8&8&9&9&9&9&9&9&10&10&10&10&10&11&11&11&11&11&11&12&12\\
		70&8&8&8&8&8&8&9&9&9&9&9&9&10&10&10&10&10&10&11&11&11&11&11&12&12\\
		71&8&8&8&8&8&8&9&9&9&9&9&9&10&10&10&10&10&10&11&11&11&11&11&11&12\\
		72&8&8&8&8&8&8&8&9&9&9&9&9&9&10&10&10&10&10&10&11&11&11&11&11&12\\
		73&7&8&8&8&8&8&8&9&9&9&9&9&9&10&10&10&10&10&10&11&11&11&11&11&11\\
		74&7&8&8&8&8&8&8&9&9&9&9&9&9&10&10&10&10&10&10&11&11&11&11&11&11\\
		75&7&8&8&8&8&8&8&8&9&9&9&9&9&9&10&10&10&10&10&10&11&11&11&11&11\\
		76&7&7&8&8&8&8&8&8&9&9&9&9&9&9&10&10&10&10&10&10&11&11&11&11&11\\
		77&7&7&8&8&8&8&8&8&9&9&9&9&9&9&9&10&10&10&10&10&10&11&11&11&11\\
		78&7&7&8&8&8&8&8&8&8&9&9&9&9&9&9&10&10&10&10&10&10&11&11&11&11\\
		79&7&7&7&8&8&8&8&8&8&9&9&9&9&9&9&10&10&10&10&10&10&11&11&11&11\\
		80&7&7&7&8&8&8&8&8&8&9&9&9&9&9&9&9&10&10&10&10&10&10&11&11&11\\
		81&7&7&7&8&8&8&8&8&8&8&9&9&9&9&9&9&10&10&10&10&10&10&11&11&11\\
		82&7&7&7&7&8&8&8&8&8&8&9&9&9&9&9&9&10&10&10&10&10&10&11&11&11\\
		83&7&7&7&7&8&8&8&8&8&8&9&9&9&9&9&9&9&10&10&10&10&10&10&11&11\\
		84&7&7&7&7&8&8&8&8&8&8&8&9&9&9&9&9&9&10&10&10&10&10&10&11&11\\
		85&7&7&7&7&8&8&8&8&8&8&8&9&9&9&9&9&9&10&10&10&10&10&10&10&11\\
		86&7&7&7&7&7&8&8&8&8&8&8&9&9&9&9&9&9&9&10&10&10&10&10&10&11\\
		87&7&7&7&7&7&8&8&8&8&8&8&8&9&9&9&9&9&9&10&10&10&10&10&10&11\\
		88&7&7&7&7&7&8&8&8&8&8&8&8&9&9&9&9&9&9&10&10&10&10&10&10&10\\
		89&7&7&7&7&7&7&8&8&8&8&8&8&9&9&9&9&9&9&9&10&10&10&10&10&10\\
		90&7&7&7&7&7&7&8&8&8&8&8&8&8&9&9&9&9&9&9&10&10&10&10&10&10\\
		91&7&7&7&7&7&7&8&8&8&8&8&8&8&9&9&9&9&9&9&10&10&10&10&10&10\\
		92&7&7&7&7&7&7&8&8&8&8&8&8&8&9&9&9&9&9&9&9&10&10&10&10&10\\
		93&7&7&7&7&7&7&7&8&8&8&8&8&8&8&9&9&9&9&9&9&10&10&10&10&10\\
		94&7&7&7&7&7&7&7&8&8&8&8&8&8&8&9&9&9&9&9&9&10&10&10&10&10\\
		95&7&7&7&7&7&7&7&8&8&8&8&8&8&8&9&9&9&9&9&9&9&10&10&10&10\\
		96&7&7&7&7&7&7&7&8&8&8&8&8&8&8&9&9&9&9&9&9&9&10&10&10&10\\
		97&7&7&7&7&7&7&7&7&8&8&8&8&8&8&8&9&9&9&9&9&9&10&10&10&10\\
		98&6&7&7&7&7&7&7&7&8&8&8&8&8&8&8&9&9&9&9&9&9&9&10&10&10\\
		99&6&7&7&7&7&7&7&7&8&8&8&8&8&8&8&9&9&9&9&9&9&9&10&10&10\\
		100&6&7&7&7&7&7&7&7&7&8&8&8&8&8&8&8&9&9&9&9&9&9&10&10&10\\
		
		\hline
	\end{tabular}
\end{table}

\begin{table}[ht]
	\hspace{-6cm}
	\caption{$r(n+k,n)$ for $2 \le k \le 100, 77\le n \le 100$} 
	\centering 
	\fontsize{6}{6}\selectfont
	\begin{tabular}{c|ccccccccccccccccccccccccc}
		\hline\hline
		k/n&77&78&79&80&81&82&83&84&85&86&87&88&89&90&91&92&93&94&95&96&97&98&99&100\\
		\hline
		2&36&37&37&38&38&39&39&40&40&41&41&42&42&43&43&44&44&45&45&45&46&46&47&47\\
		3&34&35&35&35&36&36&37&37&38&38&39&39&40&40&40&41&41&42&42&43&43&44&44&45\\
		4&32&33&33&34&34&35&35&35&36&36&37&37&38&38&38&39&39&40&40&41&41&42&42&42\\
		5&31&31&32&32&33&33&33&34&34&35&35&36&36&36&37&37&38&38&39&39&39&40&40&41\\
		6&30&30&30&31&31&32&32&32&33&33&34&34&35&35&35&36&36&37&37&37&38&38&39&39\\
		7&28&29&29&30&30&30&31&31&32&32&33&33&33&34&34&35&35&35&36&36&37&37&37&38\\
		8&28&28&28&29&29&29&30&30&31&31&31&32&32&33&33&33&34&34&35&35&35&36&36&37\\
		9&27&27&27&28&28&29&29&29&30&30&30&31&31&32&32&32&33&33&34&34&34&35&35&36\\
		10&26&26&27&27&27&28&28&28&29&29&30&30&30&31&31&32&32&32&33&33&33&34&34&35\\
		11&25&26&26&26&27&27&27&28&28&28&29&29&30&30&30&31&31&31&32&32&33&33&33&34\\
		12&25&25&25&26&26&26&27&27&27&28&28&28&29&29&30&30&30&31&31&31&32&32&33&33\\
		13&24&24&25&25&25&26&26&26&27&27&27&28&28&29&29&29&30&30&30&31&31&31&32&32\\
		14&23&24&24&24&25&25&25&26&26&26&27&27&28&28&28&29&29&29&30&30&30&31&31&31\\
		15&23&23&23&24&24&25&25&25&26&26&26&27&27&27&28&28&28&29&29&29&30&30&30&31\\
		16&22&23&23&23&24&24&24&25&25&25&26&26&26&27&27&27&28&28&28&29&29&29&30&30\\
		17&22&22&23&23&23&23&24&24&24&25&25&25&26&26&26&27&27&28&28&28&29&29&29&30\\
		18&21&22&22&22&23&23&23&24&24&24&25&25&25&26&26&26&27&27&27&28&28&28&29&29\\
		19&21&21&22&22&22&23&23&23&24&24&24&25&25&25&26&26&26&26&27&27&27&28&28&28\\
		20&21&21&21&22&22&22&22&23&23&23&24&24&24&25&25&25&26&26&26&27&27&27&28&28\\
		21&20&21&21&21&21&22&22&22&23&23&23&24&24&24&25&25&25&26&26&26&27&27&27&28\\
		22&20&20&20&21&21&21&22&22&22&23&23&23&24&24&24&25&25&25&25&26&26&26&27&27\\
		23&20&20&20&20&21&21&21&22&22&22&23&23&23&23&24&24&24&25&25&25&26&26&26&27\\
		24&19&20&20&20&20&21&21&21&22&22&22&23&23&23&23&24&24&24&25&25&25&26&26&26\\
		25&19&19&20&20&20&20&21&21&21&22&22&22&22&23&23&23&24&24&24&25&25&25&26&26\\
		26&19&19&19&19&20&20&20&21&21&21&22&22&22&22&23&23&23&24&24&24&25&25&25&25\\
		27&18&19&19&19&19&20&20&20&21&21&21&22&22&22&22&23&23&23&24&24&24&24&25&25\\
		28&18&18&19&19&19&19&20&20&20&21&21&21&21&22&22&22&23&23&23&24&24&24&24&25\\
		29&18&18&18&19&19&19&19&20&20&20&21&21&21&21&22&22&22&23&23&23&24&24&24&24\\
		30&18&18&18&18&19&19&19&19&20&20&20&21&21&21&21&22&22&22&23&23&23&23&24&24\\
		31&17&18&18&18&18&19&19&19&20&20&20&20&21&21&21&21&22&22&22&23&23&23&23&24\\
		32&17&17&18&18&18&18&19&19&19&20&20&20&20&21&21&21&21&22&22&22&23&23&23&23\\
		33&17&17&17&18&18&18&18&19&19&19&20&20&20&20&21&21&21&21&22&22&22&23&23&23\\
		34&17&17&17&17&18&18&18&19&19&19&19&20&20&20&20&21&21&21&21&22&22&22&23&23\\
		35&16&17&17&17&17&18&18&18&19&19&19&19&20&20&20&20&21&21&21&21&22&22&22&23\\
		36&16&17&17&17&17&18&18&18&18&19&19&19&19&20&20&20&20&21&21&21&22&22&22&22\\
		37&16&16&17&17&17&17&18&18&18&18&19&19&19&19&20&20&20&20&21&21&21&22&22&22\\
		38&16&16&16&17&17&17&17&18&18&18&18&19&19&19&19&20&20&20&20&21&21&21&22&22\\
		39&16&16&16&16&17&17&17&17&18&18&18&18&19&19&19&19&20&20&20&21&21&21&21&22\\
		40&16&16&16&16&17&17&17&17&17&18&18&18&19&19&19&19&20&20&20&20&21&21&21&21\\
		41&15&16&16&16&16&17&17&17&17&18&18&18&18&19&19&19&19&20&20&20&20&21&21&21\\
		42&15&15&16&16&16&16&17&17&17&17&18&18&18&18&19&19&19&19&20&20&20&20&21&21\\
		43&15&15&15&16&16&16&16&17&17&17&17&18&18&18&18&19&19&19&19&20&20&20&20&21\\
		44&15&15&15&16&16&16&16&17&17&17&17&17&18&18&18&18&19&19&19&19&20&20&20&20\\
		45&15&15&15&15&16&16&16&16&17&17&17&17&18&18&18&18&19&19&19&19&20&20&20&20\\
		46&15&15&15&15&15&16&16&16&16&17&17&17&17&18&18&18&18&19&19&19&19&20&20&20\\
		47&14&15&15&15&15&16&16&16&16&17&17&17&17&17&18&18&18&18&19&19&19&19&20&20\\
		48&14&15&15&15&15&15&16&16&16&16&17&17&17&17&18&18&18&18&18&19&19&19&19&20\\
		49&14&14&15&15&15&15&15&16&16&16&16&17&17&17&17&18&18&18&18&19&19&19&19&20\\
		50&14&14&14&15&15&15&15&16&16&16&16&16&17&17&17&17&18&18&18&18&19&19&19&19\\
		51&14&14&14&15&15&15&15&15&16&16&16&16&17&17&17&17&17&18&18&18&18&19&19&19\\
		52&14&14&14&14&15&15&15&15&16&16&16&16&16&17&17&17&17&18&18&18&18&19&19&19\\
		53&14&14&14&14&15&15&15&15&15&16&16&16&16&17&17&17&17&17&18&18&18&18&19&19\\
		54&14&14&14&14&14&15&15&15&15&15&16&16&16&16&17&17&17&17&17&18&18&18&18&19\\
		55&13&14&14&14&14&14&15&15&15&15&16&16&16&16&16&17&17&17&17&18&18&18&18&18\\
		56&13&13&14&14&14&14&15&15&15&15&15&16&16&16&16&17&17&17&17&17&18&18&18&18\\
		57&13&13&14&14&14&14&14&15&15&15&15&16&16&16&16&16&17&17&17&17&17&18&18&18\\
		58&13&13&13&14&14&14&14&15&15&15&15&15&16&16&16&16&16&17&17&17&17&18&18&18\\
		59&13&13&13&14&14&14&14&14&15&15&15&15&15&16&16&16&16&17&17&17&17&17&18&18\\
		60&13&13&13&13&14&14&14&14&14&15&15&15&15&16&16&16&16&16&17&17&17&17&18&18\\
		61&13&13&13&13&14&14&14&14&14&15&15&15&15&15&16&16&16&16&17&17&17&17&17&18\\
		62&13&13&13&13&13&14&14&14&14&14&15&15&15&15&16&16&16&16&16&17&17&17&17&17\\
		63&13&13&13&13&13&14&14&14&14&14&15&15&15&15&15&16&16&16&16&16&17&17&17&17\\
		64&12&13&13&13&13&13&14&14&14&14&14&15&15&15&15&15&16&16&16&16&17&17&17&17\\
		65&12&13&13&13&13&13&14&14&14&14&14&15&15&15&15&15&16&16&16&16&16&17&17&17\\
		66&12&12&13&13&13&13&13&14&14&14&14&14&15&15&15&15&15&16&16&16&16&17&17&17\\
		67&12&12&13&13&13&13&13&14&14&14&14&14&15&15&15&15&15&16&16&16&16&16&17&17\\
		68&12&12&12&13&13&13&13&13&14&14&14&14&14&15&15&15&15&15&16&16&16&16&16&17\\
		69&12&12&12&13&13&13&13&13&14&14&14&14&14&15&15&15&15&15&16&16&16&16&16&17\\
		70&12&12&12&12&13&13&13&13&13&14&14&14&14&14&15&15&15&15&15&16&16&16&16&16\\
		71&12&12&12&12&13&13&13&13&13&14&14&14&14&14&15&15&15&15&15&16&16&16&16&16\\
		72&12&12&12&12&13&13&13&13&13&13&14&14&14&14&14&15&15&15&15&15&16&16&16&16\\
		73&12&12&12&12&12&13&13&13&13&13&14&14&14&14&14&15&15&15&15&15&16&16&16&16\\
		74&12&12&12&12&12&13&13&13&13&13&13&14&14&14&14&14&15&15&15&15&15&16&16&16\\
		75&12&12&12&12&12&12&13&13&13&13&13&14&14&14&14&14&15&15&15&15&15&16&16&16\\
		76&11&12&12&12&12&12&13&13&13&13&13&13&14&14&14&14&14&15&15&15&15&15&16&16\\
		77&11&12&12&12&12&12&12&13&13&13&13&13&14&14&14&14&14&15&15&15&15&15&15&16\\
		78&11&11&12&12&12&12&12&13&13&13&13&13&13&14&14&14&14&14&15&15&15&15&15&16\\
		79&11&11&12&12&12&12&12&12&13&13&13&13&13&14&14&14&14&14&15&15&15&15&15&15\\
		80&11&11&12&12&12&12&12&12&13&13&13&13&13&13&14&14&14&14&14&15&15&15&15&15\\
		81&11&11&11&12&12&12&12&12&12&13&13&13&13&13&14&14&14&14&14&15&15&15&15&15\\
		82&11&11&11&12&12&12&12&12&12&13&13&13&13&13&14&14&14&14&14&14&15&15&15&15\\
		83&11&11&11&11&12&12&12&12&12&13&13&13&13&13&13&14&14&14&14&14&15&15&15&15\\
		84&11&11&11&11&12&12&12&12&12&12&13&13&13&13&13&14&14&14&14&14&14&15&15&15\\
		85&11&11&11&11&12&12&12&12&12&12&13&13&13&13&13&13&14&14&14&14&14&15&15&15\\
		86&11&11&11&11&11&12&12&12&12&12&12&13&13&13&13&13&14&14&14&14&14&14&15&15\\
		87&11&11&11&11&11&12&12&12&12&12&12&13&13&13&13&13&13&14&14&14&14&14&15&15\\
		88&11&11&11&11&11&11&12&12&12&12&12&13&13&13&13&13&13&14&14&14&14&14&14&15\\
		89&11&11&11&11&11&11&12&12&12&12&12&12&13&13&13&13&13&13&14&14&14&14&14&15\\
		90&11&11&11&11&11&11&12&12&12&12&12&12&13&13&13&13&13&13&14&14&14&14&14&14\\
		91&10&11&11&11&11&11&11&12&12&12&12&12&12&13&13&13&13&13&14&14&14&14&14&14\\
		92&10&11&11&11&11&11&11&12&12&12&12&12&12&13&13&13&13&13&13&14&14&14&14&14\\
		93&10&11&11&11&11&11&11&11&12&12&12&12&12&12&13&13&13&13&13&14&14&14&14&14\\
		94&10&10&11&11&11&11&11&11&12&12&12&12&12&12&13&13&13&13&13&13&14&14&14&14\\
		95&10&10&11&11&11&11&11&11&12&12&12&12&12&12&13&13&13&13&13&13&14&14&14&14\\
		96&10&10&10&11&11&11&11&11&11&12&12&12&12&12&12&13&13&13&13&13&13&14&14&14\\
		97&10&10&10&11&11&11&11&11&11&12&12&12&12&12&12&13&13&13&13&13&13&14&14&14\\
		98&10&10&10&11&11&11&11&11&11&12&12&12&12&12&12&12&13&13&13&13&13&14&14&14\\
		99&10&10&10&10&11&11&11&11&11&11&12&12&12&12&12&12&13&13&13&13&13&13&14&14\\
		100&10&10&10&10&11&11&11&11&11&11&12&12&12&12&12&12&13&13&13&13&13&13&14&14\\
		\hline
	\end{tabular}
\end{table}
\end{document}